\documentclass[letterpaper,11pt]{article}

\pdfoutput=1

\usepackage{algorithm}
\usepackage{float}
\newfloat{algorithm}{t}{lop}

\usepackage{amsmath,amsfonts,amssymb,bbm} 

\usepackage{amsmath, amssymb}
\usepackage{color}
\usepackage[margin=1in]{geometry}
\usepackage[noend]{algpseudocode}
\usepackage{tikz}
\usepackage{enumerate}
\usepackage{hyperref}
\usepackage{graphicx}
\usepackage{caption}
\usepackage{subcaption}
\usepackage{zerosum,csquotes}
\usepackage{linegoal}
\usepackage[inline]{enumitem}
\usepackage{comment}	
\usepackage{ctable}					
\usepackage{mathtools}
\usepackage[flushleft]{threeparttable}
\usepackage{empheq}
\usepackage{newfloat}
\usepackage{soul}

\usepackage{footnote}
\makesavenoteenv{tabular}
\makesavenoteenv{table}

\newtheorem{theorem}{Theorem}[section]
\newtheorem{lemma}[theorem]{Lemma}

\newtheorem{corollary}[theorem]{Corollary}

\newtheorem{definition}[theorem]{Definition}
\makeatletter
\def\GrabProofArgument[#1]{ #1: \egroup\ignorespaces}
\def\proof{\noindent\textbf\bgroup Proof%
	\@ifnextchar[{\GrabProofArgument}{. \egroup\ignorespaces}}

\makeatother

\newcommand{\xhdr}[1]{\vspace{2mm} \noindent{\bf #1}}
\newcommand{\OMIT}[1]{}

\usepackage{color-edits}
\addauthor{sb}{purple}    
\addauthor{md}{magenta}  
\addauthor{mt}{red}      
\addauthor{as}{blue}     

\newcommand{\discreteProblem}[0]{\ensuremath{\mathcal{DSG}}}
\newcommand{\continuousProblem}[0]{\ensuremath{\mathcal{CSG}}}

\newcommand{\M}[0]{\ensuremath{M}}

\newcommand{\Mset}{\ensuremath{\lbrack\M\rbrack}}

\renewcommand{\time}[0]{\ensuremath{T}}
\newcommand{\timeset}{\ensuremath{\lbrack\time\rbrack}}

\newcommand{\target}[1]{\ensuremath{A_{#1}}}
\newcommand{\targetset}[1]{\ensuremath{\lbrack\target{#1}\rbrack}}
\newcommand{\targetpos}[2]{\ensuremath{h_{#1, #2}}}
\newcommand{\targetweight}[2]{\ensuremath{w_{#1, #2}}}

\newcommand{\K}[0]{\ensuremath{K}}
\newcommand{\Kset}{\ensuremath{\lbrack\K\rbrack}}

\newcommand{\D}[0]{\ensuremath{\Delta}} 

\newcommand{\R}[0]{\ensuremath{R}} 

\newcommand{\parof}[1]{\ensuremath{P_{#1}}}
\newcommand{\parsize}[1]{\ensuremath{n_{#1}}}
\newcommand{\intseq}[1]{\ensuremath{\mathcal{I}_{#1}}} 
\newcommand{\intseqitem}[2]{\ensuremath{\intseq{#1}\lbrack#2\rbrack}}

\newcommand{\intseqcons}[3]{\ensuremath{\intseq{#1}\lbrack#2:#3\rbrack}}

\newcommand{\getintervals}[0]{\textproc{GetIntervals}}
\newcommand{\getintervalpoints}[0]{\textproc{GetIntervalPoints}}
\newcommand{\epsdistance}[0]{\ensuremath{\epsilon}}

\newcommand{\vertexset}[1]{\ensuremath{V(#1)}}
\newcommand{\edgeset}[1]{\ensuremath{E(#1)}}

\newcommand{\graph}[1]{\ensuremath{G_{#1}}}
\newcommand{\vertex}[3]{\ensuremath{V_{#1}\lbrack#2, #3\rbrack}}
\newcommand{\vertexcons}[4]{\ensuremath{V_{#1}\lbrack#2; #3:#4\rbrack}}
\newcommand{\source}[1]{\ensuremath{S_{#1}}} 
\newcommand{\sink}[1]{\ensuremath{S'_{#1}}} 
\newcommand{\edgeweight}[2]{\ensuremath{w_{#1,#2}}}
\newcommand{\edgecover}[2]{\ensuremath{c_{#1,#2}}}
\newcommand{\edgeflow}[2]{\ensuremath{#1(#2)}}

\newcommand{\intervalscover}[4]{\ensuremath{\mathcal{C}_{#1}(#2, #3; #4)}}
\newcommand{\inflow}[2]{\ensuremath{#1^-(#2)}}
\newcommand{\outflow}[2]{\ensuremath{#1^+(#2)}}

\DeclareMathOperator{\feas}{feas}
\newcommand{\outfeas}[2]{\ensuremath{\feas_{#1}(#2)}}

\newcommand{\sizeofflowpath}[1]{|#1|}

\newsavebox{\ieeealgbox}

\newcommand{\order}[1]{\ensuremath{\mathcal{O}(#1)}}

\DeclareFloatingEnvironment[fileext=lop, name=Algorithm]{alg}

\DeclareFloatingEnvironment[fileext=lop, name=Linear Program]{lp}

\newcounter{proccnt}

\author{
	Soheil Behnezhad\thanks{Department of Computer Science, University of Maryland. Email: \texttt{\{soheil,mahsaa,hajiagha\}@cs.umd.edu}. Supported in part by NSF CAREER award CCF-1053605,  NSF BIGDATA grant IIS-1546108, NSF AF:Medium grant CCF-1161365,
DARPA GRAPHS/AFOSR grant FA9550-12-1-0423, and another DARPA SIMPLEX grant.
}
\and Mahsa Derakhshan\footnotemark[1]
\and MohammadTaghi Hajiaghayi\footnotemark[1]
\and Aleksandrs Slivkins\thanks{Microsoft Research. Email: \texttt{slivkins@microsoft.com}.}
}

\title{A Polynomial Time Algorithm for Spatio-Temporal \mbox{Security Games}}

\begin{document}

\date{}
\maketitle

\begin{abstract}
An ever-important issue is protecting infrastructure and other valuable targets from a range of threats from vandalism to theft to piracy to terrorism. The ``defender" can rarely afford the needed resources for a 100\% protection. Thus, the key question is, \emph{how to provide the best protection using the limited available resources}.

We study a practically important class of security games that is played out in space and time, with targets and ``patrols" moving on a real line. A central open question here is whether the Nash equilibrium (i.e., the minimax strategy of the defender) can be computed in polynomial time. We resolve this question in the affirmative. Our algorithm runs in time polynomial in the input size, and only polylogarithmic in the number of possible patrol locations $(\M{})$. Further, we provide a \emph{continuous} extension in which patrol locations can take arbitrary real values.
Prior work obtained polynomial-time algorithms only under a substantial assumption, e.g., a constant number of rounds. Further, all these algorithms have running times polynomial in $\M{}$, which can be very large.
\end{abstract}

%

\maketitle

\section{Introduction}

Protecting infrastructure and other valuable targets from a range of threats from vandalism to theft to piracy to terrorism is an ever-important issue around the world, aggravated recently by increased threats of piracy and  terrorism. Providing 100\% protection usually requires more money or other resources than the ``defender" can commit. Thus, the key question is, \emph{how to provide the best protection using the limited resources that are available}.

A successful recent approach casts this issue in game-theoretic terms, modeling it as a \emph{security game}: a zero-sum game between the \emph{defender} who has some targets to protect, and the \emph{attacker} who strives to inflict damage on these targets. Usually the defender needs to commit to a particular allocation of resources, such as the schedule of patrols, whereas the attacker can strike at will; this corresponds to a classic game-theoretic model called a \emph{Stackelberg game}. The defender can (and should) randomize, e.g. so as to prevent the attacker from exploiting a particular gap in the patrol schedule. The attacker can be strategic and optimize his attack according to his beliefs about the defender's strategy. The literature has mostly adopted a pessimistic view, in which the attack is the exact best response to the defender's actual strategy. Thus, the defender's goal is to use an optimal (minimax) strategy.

This approach has resulted in a flurry of research activity, including several awards and nominations. Further, it has been adopted in a number of real-world deployments, ranging from patrol boats to airport checkpoints to US air marshals to an urban transit system to wildlife protection. (Many of them have been recognized with commendations and awards.) Other potential applications include protecting aid convoys in unstable regions, and protecting ships from piracy.

Most applications of security games are \emph{spatio-temporal} in the sense that the patrols move from one location to another with a limited speed, and only protect targets that are sufficiently close. Then the defender's strategy is a rather complicated object: a pure strategy should specify a trajectory for every patrol, possibly choosing from a very large number of possible locations. Further, in many applications targets have their own trajectories that need to be taken into consideration.

The first-order question in spatio-temporal security games is computing the equilibrium, i.e., the minimax strategy for the defender. More specifically, we focus on exact equilibrium computation in polynomial time. The prior work (e.g., \cite{fang2013optimal,bovsansky2011computing,xu2014solving}), as well as the present paper, considers a one-dimensional space (that is, patrol and target locations are on the real line), and discretizes it uniformly into the possible patrol locations.
The relevant parameters are: the number of patrols ($\K{}$), the number of targets ($\target{}$), the number of rounds of scheduling ($\time$), and the number of possible patrol locations ($\M{}$). The input specifies trajectories of targets and their values; the trajectories may be arbitrary, and the values  may change over time. Thus, the input size is
    $O(\time{}\cdot\target{}+\log(\K{}\cdot \M{}))$.
What makes the problem particularly challenging is that the number of pure strategies --- tuples of patrol trajectories --- is as large as $(\M{})^{\K{} \time{}}$.

A central open question here is whether the Nash equilibrium (i.e., the minimax strategy of the defender) can be computed in polynomial time. We resolve this question in the affirmative: our main result is an algorithm that computes the exact equilibrium in time polynomial in the input size. In particular, the running time scales only \emph{polylogarithmically} in the number of possible patrol locations. Moreover, we provide a \emph{continuous} extension in which patrol locations can take arbitrary real values, under a mild technical assumption that the target locations are rational. The dependence on the number of patrols is argued away: while a pure strategy of the defender must specify a trajectory for each patrol, we prove that $\mathtt{poly}(\time{}\target{})$ patrols suffices to protect all targets. The output is a distribution over $\mathtt{poly}(\time{}\target{})$-many pure strategies.

We improve over the state-of-art prior work \cite{fang2013optimal,xu2014solving} in several ways. First, the running time in \cite{fang2013optimal} is exponential in the number of patrols ($\K{}$) and becomes impractical even for $\K{}=3$ \cite{xu2014solving}. Second, \cite{xu2014solving} achieves a polynomial running time only under a substantial assumption: either a constant number of rounds, or that all targets have a unit value at all times, or that the ``protection ranges" of the patrols are so small that they cannot overlap for any two adjacent patrol locations. Third, the running times in \cite{fang2013optimal,xu2014solving} depend polynomially on the the number of patrol locations ($\M{}$). Finally, the polynomial running time in \cite{xu2014solving} relies on the Ellipsoid Algorithm for solving linear programs, which is notoriously slow in practice.


\xhdr{Our techniques.}
Our main algorithm works on the \emph{discretized version} of the problem, in which the possible locations for patrols are integers from $1$ to $\M{}$ (there is no such restriction on the location of targets). The algorithm consists of three parts: partitioning the spatial domain, formulating patrol placements in a single time point, and combining them to find the optimal strategy for all time points. Below we describe these three parts one by one.

First, we partition the spatial domain into a relatively small number of intervals so that the patrol locations inside each interval are ``equivalent" to one another as far as our problem is concerned. Then we use these intervals as ``atomic" patrol locations, thereby replacing the dependence on $\M{}$ with the dependence on the number of intervals.  The partitioning algorithm starts from the last round $\time{}$ and goes backwards in time: for each round $t$ it constructs a collection of intervals based on the target locations at time $t$ and the intervals constructed for time $t+1$, so as to ensure the desired ``equivalence" property. We bound the number of intervals by  $\order{\time^3\target{}}$.

\OMIT{ 
We prove this algorithm will not create more than \order{\time^3\target{}} intervals and furthermore, all the positions in the same interval have some similar characteristics. For example, let $i_t$ and $i_{t+1}$ be two intervals at time $t$ and $t+1$ respectively, then: (i) If a target $a$ is in the protection range of a position in $i_t$,  it is also in the protection range of any other position in $i_t$ (i.e., two patrols in the same interval protect the same set of targets). (ii) If there is a valid move from a position in $i_t$, to a position tin $i_{t+1}$, then for any other position in $i_t$, there exists at least one valid move to a position in $i_{t+1}$.
These characteristics make it possible to treat the intervals as the atomic possible patrol locations instead of considering the possibly huge space given in the input.
} 


Second, for each time point $t$ we construct a graph \graph{t} which models any possible \emph{snapshot} of the patrol placements at this time. More specifically, every patrol placement at time $t$ can be mapped to a specific path in \graph{t}, and every \emph{randomized} patrol placement can be mapped to a specific unit flow in \graph{t}. Furthermore, we define the cost of a path/flow in \graph{t} such that it equals the maximum utility of the attacker under the corresponding (randomized) patrol placement at time $t$, and can be computed via a linear program.


Third, we create a linear program that ``unifies" the graphs \graph{t}, and use this LP to construct the minimax strategy for the defender. The LP ensures that the (randomized) patrol placements computed in each \graph{t} are consistent with one another, in the sense that there is a valid transition from one round to the next, without violating the speed restriction. To accomplish this, the LP finds min-cost flows in each \graph{t}, and includes additional linear constraints that guarantee consistency. We post-process the solution of this LP and remove the crossing edges in the flows. Finally, we incrementally construct a mixed strategy of the defender based on the post-processed solution and prove that it is indeed the optimal strategy.

In the \emph{continuous} version of the problem, patrol locations can take arbitrary real values, and the target locations are rational.  We first re-scale all target locations to integers, and prove that this re-scaled problem instance admits a discrete solution. Then we use the algorithm from the discretized version. It is essential that the running time of the latter is polylogarithmic in $\M{}$.

\xhdr{Related work.}
Security games have been studied extensively in the past decade, see the book \cite{tambe2011security} as well as more recent work, e.g. \cite{fang2013optimal,bovsansky2011computing,Conitzer-aaai13,xu2014solving,Xu-ec16}. The research concerned both theoretical foundations as well as applications. Publicized real-world deployments include: US Coast Guard patrol boats \cite{fang2013optimal}, canine-patrol and vehicle-checkpoints scheduling in Los Angeles airport (LAX) \cite{pita2008deployed}, scheduling flights for air marshals by US Federal Air Marshal Service \cite{kiekintveld2009computing}, airport passenger screening by US Transportation Security Administration \cite{OneSize-aaai16}, fare inspection in Los Angeles transit system \cite{yin2012trusts}, and wildlife protection in Malaysia \cite{PAWS-aaai16}.

Most relevant to the present work are papers on computing minimax strategy in zero-sum spatio-temporal security games. While the initial work assumed static targets \cite{tambe2011security}, some of the later work addressed moving targets
\cite{bovsansky2011computing,fang2013optimal,xu2014solving} (as discussed above). On a related note, if the patrols are allowed to \emph{accelerate}, with an upper bound on the acceleration, then computing the defender's minimax strategy becomes NP-hard \cite{xu2014solving, Xu-ec16}.
Other work concerned solving security games that are not (necessarily) spatio-temporal or zero-sum, e.g. \cite{Conitzer-ec06,korzhyk2010complexity,xu2014solving}
A notable line of work in security games assumes that the defender does not fully know attackers' values for the targets, but can learn more about them over multiple rounds of interaction with the said attackers (see \cite{SSG-ML-survey16} for a recent survey of a subset of this work, as well as \cite{Letchford-sagt09,Marecki-aamas12,Blum-nips14,Balcan-ec15}).

In a broader game-theoretic context, our work is related to Stackelberg games and equilibrium computation. Originally introduced to model competing firms, Stackelberg games is a classic concept in game theory which appears in many textbooks and countless papers.

Computing Nash Equilibria is a central problem in algorithmic economics. While this problem is known to be PPAD-hard in general, polynomial-time algorithms exist for many natural classes of games, particularly for zero-sum games (for background, see a survey \cite{Tim-EQ-survey10} and references therein). Yet, these algorithmic results are insufficient for games in which the number of pure strategies can be exponential in the input size (see \cite{ahmadinejad2016duels, behnezhad2016faster} for examples of such games).

\xhdr{Further directions.}
Many ideas in this paper may be useful for solving other spatio-temporal security games. In particular,  the overall algorithmic framework of locally solving each ``time layer" under some compatibility constraints and then merging the ``time layers" to compute the global optimum solution appears broadly applicable.

We believe our techniques can be extended to achieve a polynomial time algorithm for several extensions of the model. In particular, we can incorporate additional constraints on the patrols, such as obstacles that the patrols cannot cross over, or speed limits that depend on a particular location. We can also handle scenarios when the spatial domain or the timeline are not evenly discretized. (For ease of presentaton, we do not include these extensions in the present paper).

A general way to model such extensions is to assume that the range of valid movements for each location is given in the input.  Whenever we still have a property that the patrols do not need to cross each other in an optimal solution (which indeed is a very natural property for homogeneous patrols), our techniques achieve a polynomial-time algorithm in the input size. However, the input size for this extended model gets large, and no longer scales polylogarithmically in \M{}.

That said, some important special cases allow for succinct input. For example, a small number of obstacles can be specified directly, rather than given implicitly via the ranges of valid movements. Designing a polynomial-time algorithm for such cases requires a problem-specific pre-processing step for partitioning the locations (which could potentially be very different from the partitioning step in this paper). However, the rest of the algorithm could be essentially the same.

It is very tempting to extend our model to a two-dimensional space. We believe some of our techniques can be useful for this extension, most importantly the compatibility constraint technique from Section~\ref{sub:all-times}. The main challenge in extending our approach is an appropriate generalization of the ``day graphs".

\section{Preliminaries}

\begin{figure}
  \centering
  \includegraphics[scale=0.70]{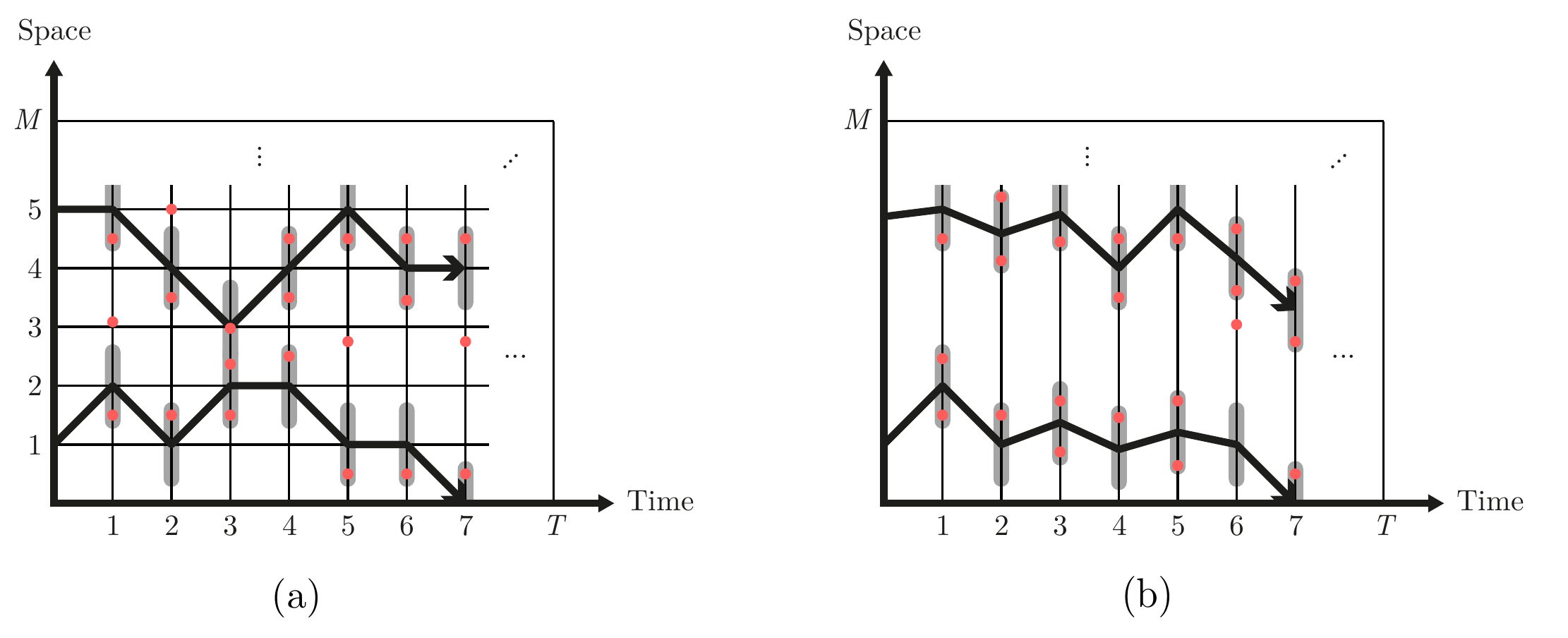}
  \caption{An illustration of how targets and patrols move in discrete (a) and continuous (b) models of the problem.}
  \label{fig:problem}
\end{figure}
Our goal is to find an optimal patrol scheduling strategy for the defender to protect a set of \target{} mobile targets in a one-dimensional space. Figure~\ref{fig:problem} illustrates the problem in a 2-D diagram; the x-axis denotes the evenly discretized temporal domain containing $\time{}+1$ time points, and the y-axis denotes the one-dimensional space of length \M{}. We say a spatial position $i$ is above $j$, if $i > j$ and similarly define a spatial position $i$ to be under $j$, if $i < j$.

The defender has \K{} homogeneous patrols to protect a set of moving targets from a potential attack. Patrols have a maximum speed of \D{}. This means, a move from a position $m_t$ at time $t$ to $m_{t+1}$ at time $t+1$ is invalid if $|m_{t+1} - m_t| > \D{}$. We consider two models of the problem: \textit{discretized model}, denoted by \discreteProblem{} (Figure~\ref{fig:problem}-a), and \textit{continuous model}, denoted by \continuousProblem{} (Figure~\ref{fig:problem}-b). In the discretized model, the position of any patrol at any time point is an integer between 0 and \M{}, but in the continuous model, the patrol locations are not restricted to be integers (note that the temporal domain is still discretized). Furthermore, for any target $a$, \targetpos{a}{t} and \targetweight{a}{t} respectively denote its position and weight at time $t$. Note that in both models, there is no restriction on the position and the speed of targets and the weight of the targets could change from time to time (e.g., ferries may not carry the same number of people at different times). The patrols protect any target within their protection radius \R{} (a fixed number for all patrols). That is, a patrol $k$ at position $m_t$ at time $t$, protects a target $a$, if $|\targetpos{a}{t} - m_t| \leq \R{}$. In Figure~\ref{fig:problem}, the grey ranges around patrols denote the area they protect. We denote the set of targets, the set of spatial positions, the set of patrols and the set of all time points by \targetset{}, \Mset{}, \Kset{} and \timeset{} respectively.

A \textit{patrol path}, is a sequence of \time{} positions $(m_1, m_2, \ldots, m_\time{})$, such that for any $t \in \timeset{}$, a move from $m_t$ to $m_{t+1}$ does not violate the speed limit (the black paths in Figure~\ref{fig:problem} denote patrol paths). A pure strategy of the defender is a set of \K{} patrol paths denoted by $\{v_k\}_{k\in\Kset{}}$. A mixed strategy of the defender, is a probability distribution over her pure strategies. A pure strategy of the attacker is a single target-time pair $(a, t)$ which means the attacker attacks target $a$ at time $t$. Let $\{v_k\}_{k\in\Kset{}}$ be the pure strategy of the defender and $(a, t)$ be the pure strategy of the attacker, attacker's utility is $0$ if target $a$ is protected by at least one patrol at time $t$ and it is \targetweight{a}{t} if it is not protected by any patrols. We assume the game is zero-sum and find minmax strategies.

Without loss of generality, we can assume $\K{} \leq \time{} \target{}$. This observation comes from the fact that with only $\time \target{}$ patrols, the defender can provide a 100\% protection without needing any more patrols. To do this, for any target-time pair $(a, t)$, the defender can put a still patrol at the location of target $a$ at time $t$.

\section{Discrete Model}
\label{sec:discrete}
The main goal of this section is to prove the following theorem:
\begin{theorem}\label{thm:main}
	There is a polynomial time (in input size) algorithm to solve \discreteProblem{}.
\end{theorem}

\subsection{Partitioning The Positions}
\label{sub:partition}

The number of pure strategies in a single time point, even with only one patrol, is not polynomial in the input size, since the number of possible locations, \M{}, could be exponentially larger than the input size. To overcome this difficulty, we partition the spatial positions into polynomially many sets of consecutive positions which we call \textit{intervals} and we only keep track of these intervals instead of maintaining the exact position of a patrol within the intervals. For example, assume there is only one target, one patrol and $\time=1$, then it only matters whether the patrol's protection range contains the location of the target or not and the exact position of the patrol does not matter.

\begin{algorithm}
\caption{Partitions the given positions into intervals}
\label{alg:partition}
\begin{algorithmic}[1]
\Function{GetIntervals}{}
	\For{$t = \time$ to $1$} 
		\If {$t < \time$}
	    	$\parof{t} \leftarrow $\Call{GetIntervalPoints}{t, \parof{t+1}}
	    \Else {}
	    	$\parof{t} \leftarrow $\Call{GetIntervalPoints}{t, $\emptyset$}
	    \EndIf
		\State $\intseq{t} \leftarrow$ Array()
    	\For{any two consecutive items $p_i$ and $p_{i+1}$ in \parof{t}}
    		\State $\intseq{t}$ .\Call{Insert}{$[p_i, p_{i+1})$}
    	\EndFor
  	\EndFor
  	\State \Return $\langle \intseq{1}, \ldots, \intseq{\time}\rangle$
\EndFunction
\Function{GetIntervalPoints}{$t$, \parof{t+1}}
	\State $\parof{t} \leftarrow$ SortedArray()
	\State $\parof{t}.\Call{Insert}{0}, \parof{t}.\Call{Insert}{\M{}}$
	\State $\epsilon \leftarrow $ a sufficiently small number in $\mathbb{R}^+$
	\For {each target $a \in \targetset{t}$}
		\If {$\targetpos{a}{t}-\R > 0$} \label{line:R} 
			$\parof{t}$.\Call{Insert}{$\targetpos{a}{t}-\R$}
		\EndIf
		\If {$\targetpos{a}{t}+\R < \M$} \label{line:R+1}
			$\parof{t}$.\Call{Insert}{$\targetpos{a}{t}+\R+\epsdistance$}
		\EndIf
	\EndFor
	\For{$p$ in $\parof{t+1}$}
		\If {$p-\D > 0$}
			$\parof{t}$.\Call{Insert}{$p-\D$} \label{line:+D}
		\EndIf
		\If {$p+\D < \M{}$}
			$\parof{t}$.\Call{Insert}{$p+\D$} \label{line:-D}
		\EndIf
	\EndFor
	\For{any item $p_i$ in $\parof{t}$}
		\If{interval $[p_i, p_{i+1})$ does not contain any patrol position}
			\State $\parof{t}.\Call{Remove}{p_i}$
		\EndIf
	\EndFor
	\State \Return $\parof{t}$
\EndFunction
\end{algorithmic}
\end{algorithm}

We use Algorithm~\ref{alg:partition} to partition the positions into meaningful intervals. For any given time point $t$, the function \getintervalpoints{}, generates a sorted array $\parof{t}=\langle p_1, \ldots, p_{\parsize{t}}\rangle$ of numbers that we call \textit{interval points} and \getintervals{} uses these generated interval points to partition the spatial positions of any time point $t$ to intervals:
$$\intseq{t} = \langle \lbrack p_1, p_2), \lbrack p_2, p_3), \ldots, \lbrack p_{\parsize{t}-1}, p_{\parsize{t}}) \rangle.$$
The intervals are assumed to be left-closed and right-open to simplify the calculations. We use \intseqitem{t}{i} to denote the $i$-th interval in \intseq{t} and use \intseqcons{t}{i}{j} to denote the set of consecutive intervals $\{ \intseqitem{t}{i}, \intseqitem{t}{i+1}, \ldots, \intseqitem{t}{j} \}$.

\begin{figure}
  \centering
  \includegraphics[scale=0.55]{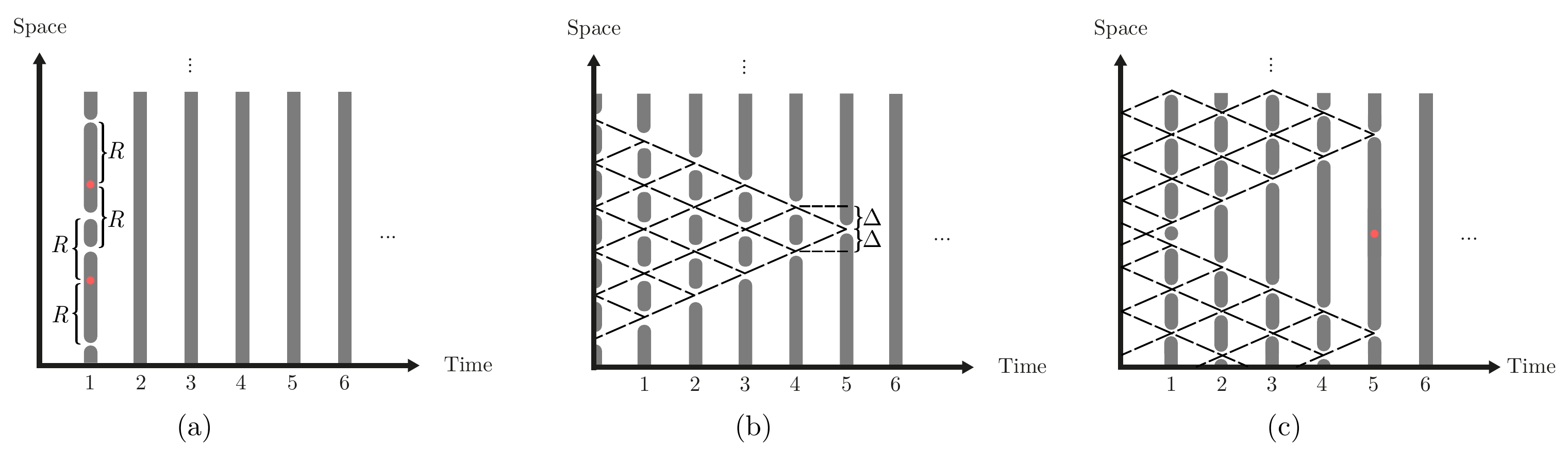}
  \caption{Interval points are added to the partitioning sets in two ways: (a) around any target there are two interval points, and (b) any interval point propagates to the previous time points. A combination of both is shown in (c).}
  \label{fig:intervals}
\end{figure}

The following lemma proves the total number of intervals is polynomial in the input size and as a corollary of that, Algorithm~\ref{alg:partition} runs in polynomial time in the input size.

\begin{lemma}\label{lem:partitionpoly}
	The total number of intervals created by Algorithm~\ref{alg:partition} is $\order{\time^3\target{}}$.
\end{lemma}
\begin{proof}
	To prove this, we charge any interval point to a target/time pair and show no target/time pair will be charged more than $\order{\time^2}$ times and since there are at most $\order{\target{}\time{}}$ target/time pairs, there will not be more than $\order{\time^3\target{}}$ interval points.	Note that there are two ways for interval points to be added to a partitioning set $\parof{t}$ (Figure~\ref{fig:intervals}):
	\begin{enumerate}
		\item For any target $a$ at time point $t$, two interval points $\targetpos{a}{t}-\R$ and $\targetpos{a}{t}+\R+\epsilon$ are added to \parof{t}. We charge these two interval points to $(a, t)$.
		\item For any interval point $p$ in $\parof{t+1}$ two interval points $p-\D$ and $p+\D$ are added to $\parof{t}$. We recursively charge these two intervals to the target/time pair that the interval point $p$ is charged to.
	\end{enumerate}
	This means an interval point $p$ could be charged to a target/time pair $(a, t)$, only if $\targetpos{a}{t} = p\pm i \D + \R + \epsilon$ or $\targetpos{a}{t} = p\pm i \D - \R$ for any integer $i$ where $i \leq t$. Although this condition is only necessary and not sufficient, but it implies at most $\order{\time^2}$ interval positions could be charged to an arbitrary target/time pair and therefore the total number of partition points is $\order{\time^3\target{}}$.
\end{proof}

Note that since by Lemma~\ref{lem:partitionpoly}, the total number of interval points at any time point $t$ is polynomial in the input size, function \getintervalpoints{}($t$, $\parof{t+1}$), which is simply a loop over $\parof{t+1}$ and the targets, runs in polynomial time.
\begin{corollary}
	Algorithm~\ref{alg:partition} halts in polynomial time in the input size.
\end{corollary}

The following lemmas prove two important properties of the partitions generated by Algorithm~\ref{alg:partition}. These properties basically imply all patrol locations within the same interval are equivalent as far as the problem is concerned.

\begin{lemma}\label{lem:sameintervalsameprotection}
	Let $k$ and $k'$ be two patrols in the same interval at any time $t$. The set of targets that $k$ and $k'$ protect at time $t$ are equal.
\end{lemma}

\begin{lemma}\label{lem:movinginintervals}
	Let $[s_i, f_i)$ and $[s_j, f_j)$ be two arbitrary intervals in $\intseq{t}$ and $\intseq{t+1}$ respectively. If there exists a feasible move from an arbitrary position in $[s_i, f_i)$ to a position in $[s_j, f_j)$, for any position in $[s_i, f_i)$, there exists a feasible move to a position in $[s_j, f_j)$.
\end{lemma}

We can now define the feasible set of a set of consecutive intervals:
\begin{definition}[feasible sets]\label{def:feas}
	We define the feasible set of \intseqcons{t}{i}{j}, denoted by $\outfeas{t}{\intseqcons{t}{i}{j}}$, to be a subset of \intseq{t+1}, containing an interval \intseqitem{t+1}{i'} iff there exists a feasible move from a position in some interval in \intseqcons{t}{i}{j} to \intseqitem{t+1}{i'}. We may occasionally abuse this notation and use the simpler form of $\outfeas{t}{\intseqitem{t}{i}}$ instead of $\outfeas{t}{\intseqcons{t}{i}{i}}$ (Figure~\ref{fig:feasset}).
	\end{definition}
\begin{figure}
  \centering
  \includegraphics[scale=0.55]{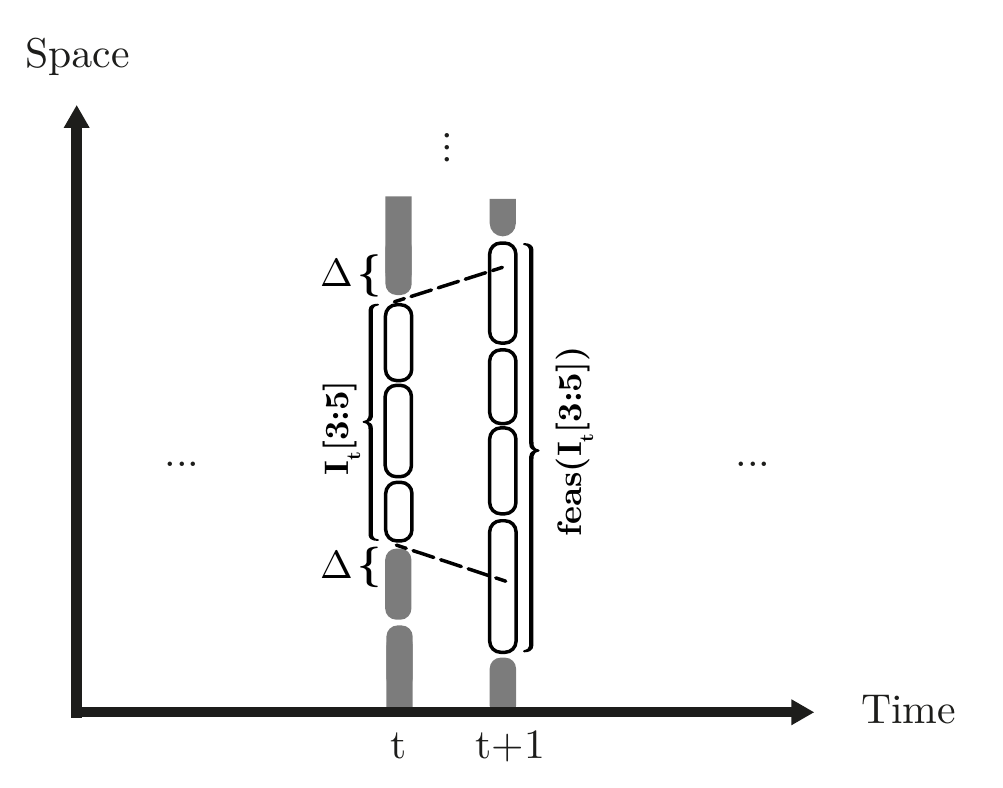}
  \caption{The white intervals at time $t+1$, are in the feasible set of the white intervals at time $t$}
  \label{fig:feasset} 
\end{figure}
It is easy to see the following corollary of Definition~\ref{def:feas}:
\begin{corollary}
	For any $\intseqcons{t}{i}{j}$ there exists a consecutive interval set $\intseqcons{t+1}{i'}{j'}$ such that
	$$\outfeas{t}{\intseqcons{t}{i}{j}} = \intseqcons{t+1}{i'}{j'}.$$
\end{corollary}

\begin{definition}[interval path] \label{def:intervalpath}
	We define an \textit{interval path} to be a sequence of $\time$ intervals $\langle \intseqitem{i}{x_i}\rangle_\time$, such that for any time point $t \in \timeset{}$, $\intseqitem{t+1}{x_{t+1}} \in \outfeas{t}{\intseqitem{t}{x_{t}}}$. Moreover, for any interval path $\xi = \langle \intseqitem{i}{x_i}\rangle_\time$, we define $S(\xi)$ to be the set of all patrol paths that are within $\xi$. More formally, a patrol path $\langle m_i \rangle_\time$ is within $S(\xi)$ if and only if for any time point $t \in \timeset$, $m_t$ is in interval $\intseqitem{t}{x_t}$.
\end{definition}

Note that any patrol path is within exactly one interval path, since intervals do not overlap and they cover all locations. It could also be obtained from the following lemma that $S(\xi)$ is never empty for an interval path $\xi$. The proof is to choose any position in $\intseqitem{1}{x_1}$ and following the valid movements until we reach a position in $\intseqitem{\time}{x_{\time}}$.
\begin{lemma}\label{lem:surjective}
	For any interval path $\xi = \langle \intseqitem{i}{x_i}\rangle_\time$, there is at least one patrol path in $S(\xi)$.
\end{lemma}

 Note that by Lemma~\ref{lem:sameintervalsameprotection}, two patrols that are in the same interval, protect the same set of targets at that specific time. This implies that the patrol paths that are within the same interval path, protect the same set of targets at all times and could be replaced with one another in any strategy, without changing the utilities. This means the amount of information encoded in an interval path is sufficient to describe the important characteristics of strategies and find the optimal one.
\subsection{Strategies In a Single Time Point}
\label{sub:single-time}

In this section we explain how to locally find the best strategy for a single time point ignoring the speed limitations. Note that although we proved the number of intervals is polynomial in the input size, there are still exponentially many different ways to place our \K{} patrols in them. This section describes how we can resolve this problem and find the best strategy.

We use the term \textit{snapshot} to denote a patrol placement at a single time point and formally define it as follows:
\begin{definition}[snapshots]
	A pure snapshot at time point $t$, is an assignment of patrols to intervals of time $t$. We denote it by a sorted sequence of \K{} intervals $\langle \intseqitem{t}{y_i} \rangle_\K$ such that for any $i \in \Kset{}$, $y_i \leq y_{i+1}$. A mixed snapshot, denoted by $\{(d_i, p_i)\}_{n}$ is a probability distribution over $n$ pure snapshots where $p_i$ denotes the probability of choosing pure snapshot $d_i$ and $\Sigma_{i=1}^{n}p_i = 1$.
\end{definition}

For any time point $t$, we construct a weighted directed graph $\graph{t}$ (called a \textit{day graph}) and give a one-to-one mapping between pure snapshots at time $t$ and paths from \source{t} (source vertex) to \sink{t} (sink vertex), where \source{t} and \sink{t} are two specific vertices of \graph{t}. Moreover, we map any mixed snapshot at time point $t$, to a network flow of 1 unit from \source{t} to \sink{t}. This mapping is not necessarily one-to-one and many mixed snapshots may be mapped to the same network flow; however, the maximum utility of the attacker in all such mixed snapshots, will be the same.

\begin{figure}
  \centering
  \includegraphics[scale=0.70]{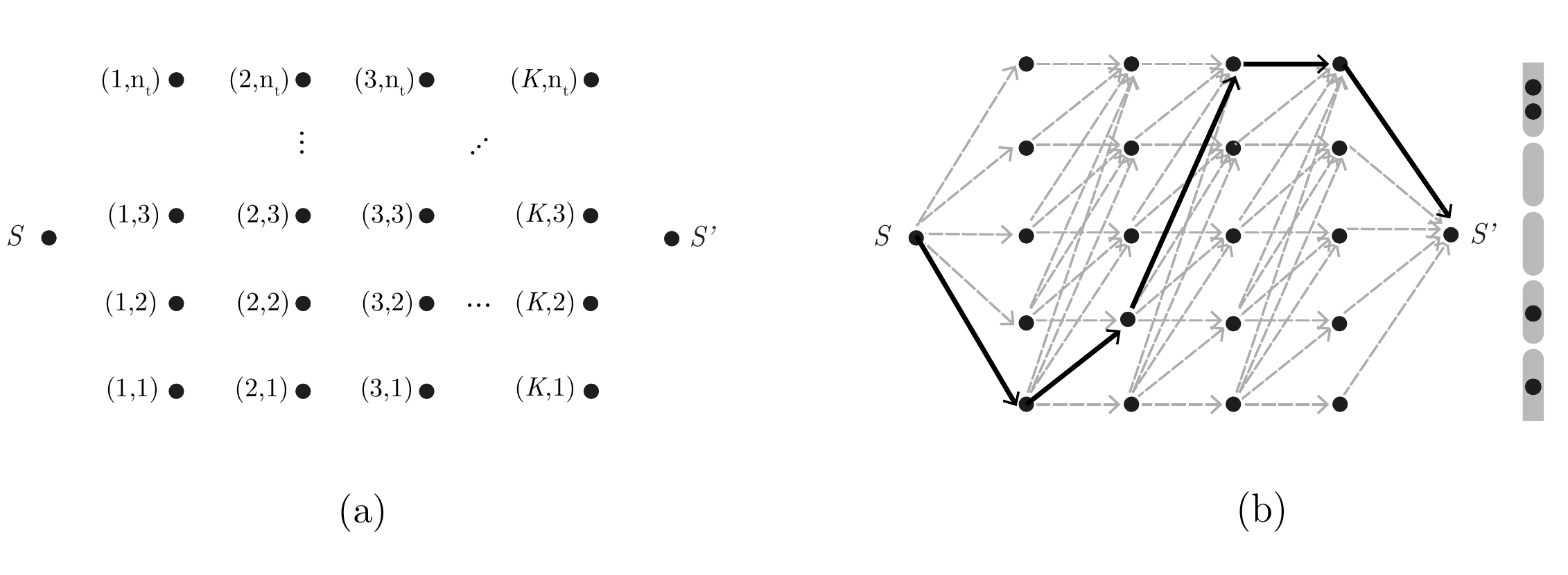}
  \caption{Figure (a) shows the vertices of a day graph and figure (b) shows a sample day graph. The highlighted path in figure (b) shows a canonical path and its equivalent patrol placement is shown in the vertical bar next to it, each grey piece denotes an intervals and the black dots denote the location of patrols in them.}
  \label{fig:graphpure}
\end{figure}

In Definition~\ref{def:graph} we formally explain how \graph{t} is constructed. An informal explanation of it is as follows: the vertex set of \graph{t}, as shown in Figure~\ref{fig:graphpure}-a, includes a vertex \source{t}, a vertex \sink{t} and a grid of $\K\times \parsize{t}$ vertices (recall that \parsize{t} is the number of intervals at time $t$) each denoted by \vertex{t}{x}{y}. There is an edge from \source{t} to any vertex in the first column of the grid \vertex{t}{1}{y}, and there is an edge from any vertex in the last column of the grid \vertex{t}{\K}{y} to \sink{t}. Also for any $x$, $y$, and $y'$, there is an edge from \vertex{t}{x}{y} to \vertex{t}{x+1}{y'} if $y \leq y'$. Furthermore, we define a \textit{canonical path} to be any path from \source{t} to \sink{t} and define a \textit{canonical flow} to be any flow of unit 1 from \source{t} to \sink{t} (Definition~\ref{def:canonical}). We give a one-to-one mapping between canonical paths in \graph{t} and pure snapshots at time point $t$ in Definition \ref{def:puremapping} and map any mixed snapshot to a canonical flow in Definition~\ref{def:mixedmapping}. Figure~\ref{fig:graphpure}-b shows a sample day graph, a canonical path in it and its equivalent pure snapshot. Moreover, we assign weights to the edges of \graph{t} such that the maximum payoff of the attacker for a pure (mixed) snapshot equals the cost of its corresponding canonical path (flow). The cost of a canonical path and a canonical flow is defined in Definition~\ref{def:canonical}.

For any target $a$ and any intervals $i$ and $i'$ at time $t$, we define the binary value $\intervalscover{t}{i}{i'}{a}$ to be $1$ if the following two conditions hold: (1) target $a$ is located in a position between $i$ and $i'$ (non-inclusive) at time $t$, and, (2) $a$ could not be protected by any patrol at any arbitrary position in $i$ or $i'$; otherwise we set $\intervalscover{t}{i}{i'}{a}$ to be 0. We similarly define two binary variables $\intervalscover{t}{\varnothing}{i}{a}$ and $\intervalscover{t}{i}{\varnothing}{a}$ for the border cases. We set $\intervalscover{t}{\varnothing}{i}{a}$ to be 1 iff target $a$ is in a position below $i$ where no patrol in $i$ can protect it and set $\intervalscover{t}{i}{\varnothing}{a}$ to be 1 iff target $a$ is in a position above $i$ where no patrol in $i$ can protect it. Assume a patrol placement $p$ does not protect a target $a$ at time $t$. Let $i$ and $i'$ be the closest intervals to target $a$, that contain at least one patrol and are below and above $a$ respectively (set them to be $\varnothing$ if no such interval exits), then by definition $\intervalscover{t}{i}{i'}{a} = 1$. Using this definition, we can now formally define a day graph and canonical path/flow.
\begin{definition}[day graph]
\label{def:graph}
	Given a time point $t$, we construct graph \graph{t} as follows:
	\begin{enumerate}
		\item Graph \graph{t} contains a vertex \source{t} (source), a vertex \sink{t} (sink) and $\K \times \parsize{t}$ other vertices, each denoted by \vertex{t}{x}{y} for $1 \leq x\leq \K$ and $1 \leq y \leq \parsize{t}$.
		\item For any $y$ such that $1 \leq y \leq \parsize{t}$, there is an edge from \source{t} to \vertex{t}{1}{y}. If $e$ denotes an edge of this kind, for any target $a$, we define $\edgecover{e}{a}$ to be $\intervalscover{t}{\varnothing}{\intseqitem{t}{y}}{a}$.
		\item For any $y$ such that $1 \leq y \leq \parsize{t}$, there is an edge from \vertex{t}{\K}{y} to \sink{t}. If $e$ denotes an edge of this kind, for any target $a$, we define $\edgecover{e}{a}$ to be $\intervalscover{t}{\intseqitem{t}{y}}{\varnothing}{a}$.
		\item For any two vertices \vertex{t}{x}{y} and \vertex{t}{x+1}{y'}, if $y \leq y'$, there is an edge from \vertex{t}{x}{y} to \vertex{t}{x+1}{y'}. If $e$ denotes this edge, for any target $a$, we define $\edgecover{e}{a} = \intervalscover{t}{\intseqitem{t}{y}}{\intseqitem{t}{y'}}{a}$.
	\end{enumerate}
\end{definition}

\begin{definition}[canonical path/flow]\label{def:canonical}
	In a day graph \graph{t} any path from \source{t} to \sink{t} is a canonical path. Let $E=\{ e_1, e_2, \ldots, e_{\K+1} \}$ be the set of edges in a canonical path, and $a$ be an arbitrary target. We define the cost of this canonical path for target $a$ to be:
	\begin{equation}
		 \sum_{e \in E} \edgecover{e}{a} . \targetweight{t}{a} 	\end{equation}
	Also, any flow of unit 1 from \source{t} to \sink{t} is a canonical flow. We denote any canonical flow with a function $f:\edgeset{\graph{t}}\rightarrow [0, 1]$, where \edgeflow{f}{e} denotes the flow passing through an edge $e$. The cost of an arbitrary canonical flow $f$, for day graph \graph{t} is:
	\begin{equation}
		\max \sum_{e \in \edgeset{\graph{t}}} \edgeflow{f}{e} . \edgecover{e}{a} . \targetweight{t}{a} \qquad \forall a; \text{ where } a \in \targetset.
	\end{equation}
\end{definition}
Intuitively speaking, \edgecover{e}{a} is 1 if and only if having edge $e$ in a canonical path $p$ implies that in the ``equivalent" pure snapshot of $p$, target $a$ is not covered by any patrol.  The formal mapping of canonical paths and flows to snapshots is as follows.

\begin{definition}[pure snapshot mapping]
\label{def:puremapping}
	Let $s_t = \langle \intseqitem{t}{y_i}\rangle_\K$ be a pure snapshot (recall that $y_i \leq y_{i+1}$ for any $i \in \Kset{}$). We map $s_t$ to the following canonical path: $$p_t = \langle \source{t}, \vertex{t}{1}{y_1}, \vertex{t}{2}{y_2}, \ldots, \vertex{t}{\K}{y_{\K}}, \sink{t} \rangle$$ and similarly map $p_t$ to $s_t$ and say $s_t$ and $p_t$ are equivalent.
\end{definition}
	
\begin{definition}[mixed snapshot mapping]
	\label{def:mixedmapping}
	Let $m=\{(d_i, p_i)\}_{n}$ be a mixed snapshot at time $t$. Also let $f_i$ denote a flow of unit $p_i$ from \source{t} to \sink{t} through the edges of the equivalent canonical path of $d_i$. We construct flow $f$ as follows: for any edge $e$ of \graph{t}, $\edgeflow{f}{e} = \Sigma_{i=1}^{n} \edgeflow{f_i}{e}$. Note that since by definition of a mixed snapshot, $\Sigma_{i=1}^{n} p_i = 1$, $f$ is a flow of unit 1 from \source{t} to \sink{t}, and hence is a canonical flow. We map $m$ to $f$.
\end{definition}
In Lemma \ref{lem:purelemma} we prove that the payoff of the attacker if he attacks target $a$ at time $t$ while the placement of patrols is represented by the pure strategy $s$, equals to the cost of the target $a$ in the canonical path equivalent to $s$. Then in Lemma \ref{lem:mixedcost} we prove that the maximum payoff of the attacker at time $t$ while the strategy of defender is represented by the mixed snapshot $m$ is equal to the cost of canonical flow equivalent to $m$. These two lemmas can be directly obtained by the given definitions, however, for space limitations, their formal proofs are left to the appendix.

\begin{lemma} \label{lem:purelemma}
	Let $s$ be a pure snapshot at time $t$, and $a$ be an arbitrary target. The payoff of the attacker with respect to $s$, if he attacks the target $a$ at time $t$, equals the cost of the target $a$ in the canonical path equivalent (Definition~\ref{def:puremapping}) to $s$.
\end{lemma}

\begin{lemma}
\label{lem:mixedcost}
	Let $r$ be a mixed snapshot at time $t$ and let $f$ denote the canonical flow that $r$ is mapped to. The maximum expected payoff of the attacker at time $t$ with respect to $r$, equals the cost of $f$.
\end{lemma}
\subsection{Best Strategy For All Time Points}
\label{sub:all-times}

Lemma~\ref{lem:mixedcost} implies if our goal is to minimize the  maximum payoff of the attacker at a single time point $t$, it suffices to find a canonical flow in \graph{t} with minimum cost. Although this works for the special case when $\time=1$, but it does not consider the movement of patrols and their speed limits. More precisely, a pure strategy for the defender could be shown as a sequence of pure snapshots $\langle s_1, s_2, \ldots, s_\time \rangle$. However there is one important condition: for any $i \in \timeset$, there must be a feasible transition from $s_i$ to $s_{i+1}$. This is also the case for mixed snapshots and two consecutive ones may not be necessarily compatible. In this section we resolve this issue and prove Theorem~\ref{thm:main}.

Our algorithm to find the optimal strategy of the defender consists of three main steps. In the first step, which is explained in more details in Section~\ref{sec:LP}, we run an LP that returns a canonical flow for each day graph \graph{1}, \ldots, \graph{\time}. Apart from the constraints to ensure we get valid canonical flows with minimum overall cost, our LP contains an extra constraint for compatibility of these canonical flows. In the second step (Section~\ref{sec:adjusting}), while keeping the overall characteristics of these canonical flows unchanged, we adjust them  in a way to make sure no two crossing edges in any of the day graphs have a positive flow. Finally, in the third step (Section~\ref{sec:construction}), we construct a mixed strategy for the defender based on the adjusted canonical flows. 

Let $s_t$ and $s_{t+1}$ denote two pure snapshots representing the placement of patrols in a valid pure strategy $p$ at two consecutive times. In the following lemma we prove that there exists a feasible move from $i$-th interval of $s_t$ to the $i$-th interval of $s_{t+1}$ (recall that intervals in pure snapshots are sorted based on their position). We prove this lemma by induction on the number of patrols. At each step we prove that there exists a feasible move from the top most interval in $s_t$ to the top most interval in $s_{t+1}$ and we prove if we match these two together and remove them, we can construct another pure strategy that contains the remaining intervals.
	
\begin{lemma}\label{lemma: noncrossingintervalpath}
If $\langle \intseqitem{t}{y^{t}_i}\rangle_\K{}$ and $\langle \intseqitem{t+1}{y^t_i}\rangle_\K{}$ are two pure snapshots at time $t$ and $t+1$ in at least one valid pure strategy $p$, then for any $j\in \Kset$ we have $\intseqitem{t+1}{y^{t+1}_j} \in \outfeas{t}{\intseqitem{t}{y^{t}_j}}$.
\end{lemma} 

In the following definition, we define what it means for a patrol path to be \textit{intervally above}, \textit{below} or \textit{equal} to another patrol path:
\begin{definition}
Let $v = \langle m_0, m_1, \ldots, m_r\rangle$ and $v' = \langle m'_0, m'_1, ..., m'_\time \rangle$ be two patrol paths. We say $v$ and $v'$ are intervally equal if for any $t \in \timeset$, $m_t$ and $m'_t$ are in the same interval. We also define $v$ to be intervally under $v'$, if for any $t \in \timeset$,  either $m_t < m'_t$ or $m_t$ and $m'_t$ are in the same interval. Similarly, we define $v$ to be intervally above $v'$, if $v'$ is intervally under $v$.
\end{definition}

Next, in Lemma~\ref{lem: sortingintervalpath} we prove that there exists an optimal strategy of the defender that for any pure strategy $p$ in its support, there is an ordering of interval paths in $p$ such that, $i$-th interval path is always intervally under $j$-th interval path if $1\leq i<j\leq \K$. To prove this we use Lemma~\ref{lemma: noncrossingintervalpath} that indicates there exists a possible move from the $k$-th interval in  pure snapshot $s_1$ to  the $k$-th interval in pure snapshot $s_2$ if $s_1$ and $s_2$ represent the patrols' placement of a pure strategy in two consecutive times and  $1 \leq k\leq \K$. For any patrol $k\in \Kset$, we construct an interval path that contains the $k$-th interval of all the pure snapshots in p, and we assign patrol $k$ to this interval path. It is easy to see that if we order the patrols from 1 to $\K$ the interval path assigned to patrol $k_1$ is under the interval path assigned to patrol $k_2$ if $1\leq k_1\leq k_2 \leq \K$. This lemma is very similar to Lemma 3 of \cite{xu2014solving} but adopted to intervals paths.

\begin{lemma} \label{lem: sortingintervalpath}
	There exists an optimal mixed strategy of the defender, such that  for every pure strategy $p$ in its support there is an ordering of interval paths $\langle\xi_1, \dots, \xi_\K  \rangle$ such that the following condition holds for this ordering:
	for any two interval paths $\xi_i$ and $\xi_j$, in the pure strategy p,  $\xi_i$ is intervally under $\xi_j$ if $i\leq j$.
\end{lemma}

Again, for space limitations, the full proof is left to the appendix. However, intuitively, starting from any given optimal solution one can swap the remaining path of any two patrols that cross each other without losing anything. This eventually resolves all crosses and gives a desired optimal solution.

Let $s$ denote an optimal strategy of the defender that satisfies the condition mentioned in Lemma~\ref{lem: sortingintervalpath}, and let $\langle\xi_{p,1}, \xi_{p,2} \dots, \xi_{p,\K}  \rangle$ denote the ordering of interval paths in pure strategy $p$ in support of $s$ such that $\xi_{p,i}$ is intervally under $\xi_{p,j}$ if $1\leq i \leq j \leq \K$. Without loss of generality we assume for any $i\in \Kset$  the same patrol is assigned  to interval path $\xi_{p,i}$ for all $p$ in support of $s$, and it is denoted by $k_i$. Therefore, if $\langle \intseqitem{t}{y_1}, \intseqitem{t}{y_2}, \dots, \intseqitem{t}{y_\K} \rangle$ denotes a pure snapshot that represents the patrols' placement in an arbitrary time point $t\in\timeset{}$ in pure strategy $p$ in support of $s$, $\intseqitem{t}{y_i}$ is the position of patrol $k_i$ in pure strategy $p$ at this time. Moreover, let $m$ denote the mixed snapshot of strategy $s$ at time $t$. The flow passing through the vertex \vertex{t}{i}{j}, denotes the probability with which patrol $k_i$ is placed in the $j$-th interval at time $t$. So, all the data related to position of patrol $k_i$ at time $t$ is in the column $i$ of the day graph of time $t$. We use this later in the paper.

\subsubsection{Linear Programming} \label{sec:LP}

In this section we explain how the first step of our algorithm, the LP, works.

Note that a flow of 1 unit in \graph{t} with minimum weight, minimizes the attacker's payoff at time $t$. To minimize the attacker's payoff at all time points, we need to minimize the cost of the canonical flow with the maximum cost. To do this, for any edge $e$ in any day graph \graph{t} we define an LP variable \edgeflow{f_t}{e} which specifies the amount of flow passing through $e$. Moreover for any time point $t$, we include the following constraints in our LP:
\begin{enumerate}
	\item For each vertex $v$ of \graph{t} (except for the source vertex \source{t} and the sink vertex \sink{t}) the amount of ingoing flow to $v$ is equal to the amount of outgoing flow from $v$.
	\item The amount of outgoing flow from \source{t} is 1.
	\item The amount of ingoing flow to \sink{t} is 1.
	\item The amount of flow passing through any edge $e$ is not negative.
	\item The cost of flow through any edge $e$ and for any target $a$ in \graph{t}, specified by $\edgeweight{e}{a} \times \edgeflow{f_t}{e}$ is not more than $u$.
\end{enumerate}
And we set the objective function of our LP to minimize $u$, which is the overall cost of canonical flows. However, as we said earlier the canonical flows we find must be compatible; therefore apart from the aforementioned constraints, we define a compatibility constraint. Recall that \outfeas{t}{\intseqcons{t}{i}{j}} denotes a collection of intervals at time $t+1$ and contains an interval $i'$, iff there is a valid move from an interval in \intseqcons{t}{i}{j} to $i'$. We define a very similar concept for day graphs:
\begin{definition}
	Let $\vertexcons{t}{x}{i}{j}$ denote the set of consecutive grid vertices $\{\vertex{t}{x}{i}, \vertex{t}{x}{i+1}, \ldots, \vertex{t}{x}{j}\}$ in \graph{t}. Recall that by definition of \intseqcons{t}{i}{j}, any vertices in $\vertexcons{t}{x}{i}{j}$ is equivalent to an interval in \intseqcons{t}{i}{j}. We define \outfeas{t}{\vertexcons{t}{x}{i}{j}} as follows:
		\outfeas{t}{\vertexcons{t}{x}{i}{j}} is a subset of grid vertices of \graph{t+1} containing a vertex \vertex{t+1}{x}{i'} if and only if \vertex{t+1}{x}{i'} is equivalent to an interval in \outfeas{t}{\intseqcons{t}{i}{j}}. 

\end{definition} 

The compatibility constraint we use in our LP is as follows: for any set of consecutive vertices \vertexcons{t}{x}{i}{j}, the amount of flow passing through the vertices in \vertexcons{t}{x}{i}{j} is not more than the amount of flow passing through the vertices in $\outfeas{t}{\vertexcons{t}{x}{i}{j}}$. Intuitively, this constraint indicates that for any set of consecutive intervals \intseqcons{t}{i}{j}, the probability that there exists a patrol in it, should not be more than the probability of having a patrol in its feasible set (\outfeas{t}{\intseqcons{t}{i}{j}}) in the next time point. Note that by definition, \outfeas{t}{\intseqcons{t}{i}{j}} contains all of the valid intervals that a patrol in \intseqcons{t}{i}{j} can move to; therefore it is obvious why this constraint is necessary. The sufficiency of this constraint to prove compatibility of snapshots, however, comes later when we explain how we construct an optimal strategy based on the adjusted LP solution. The formal definition of the LP is given in Linear Program~\ref{lp:lp}.

\begin{lp}
\begin{align}
	\label{cons:start}&\min \hspace{0.4cm} u \\[0.2cm]
	&\inflow{f_t}{v}  = \outflow{f_t}{v} \hspace{1cm} & \forall t, v: t \in \timeset, v \in \vertexset{\graph{t}}-\{\source{t}, \sink{t}\}\\ 
	\label{cons:sflow}&\outflow{f_t}{\source{t}} = 1 & \forall t: t \in \timeset \\
	\label{cons:end}&\inflow{f_t}{\sink{t}} = 1 & \forall t: t \in \timeset \\
	\label{cons:possitive}&\edgeflow{f}{e} \geq 0 & \forall e, t: t \in \timeset, e \in \graph{t} \\
	\label{cons:consteq}&\sum_{e \in \edgeset{\graph{t}}} \edgeflow{f}{e}.\edgecover{e}{a}.\targetweight{t}{a} \leq u & \forall t, a: t \in \timeset, a \in \targetset{t} \\
	\label{cons:compatibility}&\sum_{v \in \vertexcons{t}{k}{i}{j}}\outflow{f_t}{v} \leq \sum_{v' \in \outfeas{t}{\vertexcons{t}{k}{i}{j}}} \outflow{f_{t+1}}{v'} & \forall t, k, i, j: t \in \timeset, k\in \Kset, 1 \leq i \leq j \leq \parsize{t} 
\end{align}
\caption{Variable \edgeflow{f_t}{e}, which is defined for any edge $e$  in the day graph \graph{t} where $t$ could be any time point in \timeset{}, denotes the amount of flow passing through $e$. By \inflow{f_t}{v} we mean the total flow coming into vertex $v$ (i.e., $\inflow{f_t}{v} = \Sigma \edgeflow{f_t}{e}$ where $e$ is any edge ending at $v$). Also \outflow{f_t}{v} denotes the total flow coming out of vertex $v$ (i.e., $\inflow{f_t}{v} = \Sigma \edgeflow{f_t}{e}$ where $e$ is any edge starting from $v$).}
\label{lp:lp}
\end{lp}

By the end of Section~\ref{sub:all-times}, we prove the solution of LP~\ref{lp:lp} is equal to the utility of the attacker if both players play their optimal strategies. Lemma~\ref{lem:optinlp} proves a weaker claim:

\begin{lemma}\label{lem:optinlp}
The solution of Linear Program~\ref{lp:lp} gives a lower bound for the utility of the attacker when both players play their optimal strategies.
\end{lemma}

To prove Lemma~\ref{lem:optinlp}, we start from an optimal strategy of the defender satisfying the condition of Lemma~\ref{lem: sortingintervalpath} that the interval paths do not cross each other. Based on this strategy, we construct a feasible solution for LP~\ref{lp:lp} in which the value of $u$ is equal to the maximum possible utility of the attacker. Note that this only proves LP~\ref{lp:lp} gives a lower bound for the utility of the attacker when both players play their minimax strategies since the feasible solution we considered is not necessarily the optimum solution of the LP (although as we said before, we will later prove that they are exactly the same).

\newcommand{\f}[0]{\ensuremath{\mathcal{F}}}

\subsubsection{Adjusting The LP Solution} \label{sec:adjusting} In this section we give an algorithm that adjusts any optimal solution of LP~\ref{lp:lp} to resolve their ``crossing flows". We later use the adjusted solution to construct the defender's optimal strategy.

We start by defining what we mean by \textit{crossing edges} and \textit{crossing flows}:
\begin{definition}[crossing edges and crossing flow]\label{def:crossingedges}
	Let $e=(\vertex{t}{x}{y_1}, \vertex{t}{x+1}{y_2})$ and $e'=(\vertex{t}{x}{y'_1}, \vertex{t}{x+1}{y'_2})$ be two arbitrary edges of day graph \graph{t} where $y_1 < y'_1$. We say $e$ and $e'$ cross, if and only if $y_2 > y'_2$. Moreover, if $f$ is a canonical flow of $\graph{t}$, the crossing edge pair $(e, e')$ is a crossing flow in $f$, if $f(e) > 0$ and $f(e') > 0$.
\end{definition}

In the following definition, we give a total ordering on the crossing flows in a canonical flow, which we later use in the algorithm we provide to resolve them.
\begin{definition}[crossing flows' ordering] \label{def:comparison}
	Let $(e_1, e_2)$ and $(e_3, e_4)$ be two crossing flows of a canonical flow. Also let $e_1=(\vertex{t}{x}{y_1}, \vertex{t}{x+1}{y'_1})$, $e_2=(\vertex{t}{x}{y_2},\vertex{t}{x+1}{y'_2})$, $e_3=(\vertex{t}{x'}{y_3}, \vertex{t}{x'+1}{y'_3})$, and $e_4=(\vertex{t}{x'}{y_4}, \vertex{t}{x'+1}{y'_4})$ be the vertices of these edges. We say $(e_1, e_2) < (e_3, e_4)$ if and only if one of the following conditions hold:
	\begin{enumerate}
		\item $x < x'$
		\item $x = x'$ and $y_1 < y_3$.
		\item $x = x'$ and $y_1 = y_3$ and $y'_1 < y'_3$.
		\item $x = x'$ and $y_1 = y_3$ and $y'_1 = y'_3$ and $y'_2 < y'_4$.
		\item $x = x'$ and $y_1 = y_3$ and $y'_1 = y'_3$ and $y'_2 = y'_4$ and $y_2 < y_4$.
	\end{enumerate}
\end{definition}

%

Algorithm~\ref{alg:cross} is the formal pseudo-code of how we resolve all crossing flows. At each step, the algorithm resolves the minimum crossing flow (minimum based on the total ordering defined in Definition~\ref{def:comparison}) and continues this process until there is no other one. Figure~\ref{fig:non-crossing} illustrates how a single crossing flow is resolved.

\begin{algorithm}
\caption{Resolves crossing flows}
\label{alg:cross}
\begin{algorithmic}[1]
\Function{ResolveCrosses}{$f_1$, $\graph{1}$, $f_2$, $\graph{2}$, \ldots, $f_{\time}$, $\graph{t}$} \label{func: crosses}
	\For {each timepoint $t \in \timeset$}
		\While{$f_t$ contains any crossing flow}
			\State \{(\vertex{t}{x}{y_1}, \vertex{t}{x+1}{y_2}), (\vertex{t}{x}{y'_1}, \vertex{t}{x+1}{y'_2})\} = \Call{FindNextCross}{\graph{t}, $f_t$}
			\State $e \leftarrow (\vertex{t}{x}{y_1}, \vertex{t}{x+1}{y_2})$
			\State $e' \leftarrow (\vertex{t}{x}{y'_1}, \vertex{t}{x+1}{y'_2})$
			\State \Call{ResolveCross}{$e$, $e'$}
		\EndWhile
	\EndFor
\EndFunction

\Function{FindNextCross}{\graph{t}, $f_t$}
	\State $(e, e') \leftarrow$ find the minimum crossing flow in $f_t$
	\State \Return $(e, e')$
\EndFunction

\Function{ResolveCross}{$e$, $e'$} \label{func:cross}
		\State $f_m \leftarrow min(f_t(e), f_t(e'))$
		
		\State $f_t(e') \leftarrow f_t(e') - f_m$
			
		\State $f_t(e) \leftarrow f_t(e) - f_m$
		
		\State $(\vertex{t}{x}{y_1}, \vertex{t}{x+1}{y_2})\leftarrow e $
		\State $(\vertex{t}{x}{y'_1}, \vertex{t}{x+1}{y'_2})\leftarrow e' $
		\State $f_t((\vertex{t}{x}{y_1}, \vertex{t}{x+1}{y'_2})) \leftarrow f_m$
		\State $f_t((\vertex{t}{x}{y'_1}, \vertex{t}{x+1}{y_2})) \leftarrow f_m$
\EndFunction
\end{algorithmic}
\end{algorithm}

\begin{figure}
  \centering
  \includegraphics[scale=0.50]{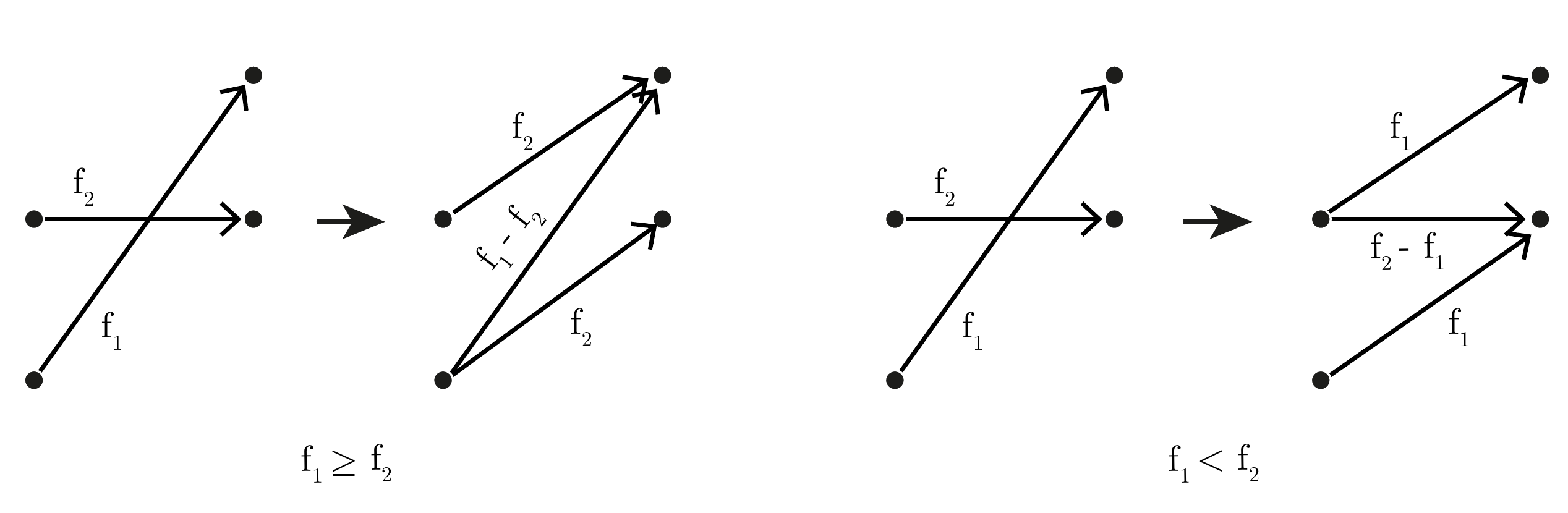}
  \caption{The whole idea of how to locally resolve the issue of having a crossing flow. The details are formally explained in the text.}
  \label{fig:non-crossing}
\end{figure}

\begin{lemma} \label{lem:crosscost}
	Function \textproc{ResolveCrosses} of Algorithm \ref{alg:cross} does not increase the total cost of the input canonical flow.
\end{lemma}

To prove Lemma~\ref{lem:crosscost}, we consider all different locations of targets for which changing the flows might affect the utility of players and prove in none of these cases the total cost is increased. The complete proof is left to the appendix.

\begin{lemma}  \label{lem:crosstime}
	The running time of the function \textproc{ResolveCrosses} in the Algorithm \ref{alg:cross} is polynomial.
\end{lemma}

The proof scheme of Lemma~\ref{lem:crosstime} is to show the minimum crossing flow at step $i$ is strictly less than the minimum crossing flow at step $i+1$ (after the previously minimum crossing flow is resolved). Consequently, since the total number of possible crossing edges of a day graph is polynomial, the number of steps until the algorithm halts is polynomial. Again, we left the formal proof of this lemma to the appendix for space limitations.

Note that another property of Algorithm~\ref{alg:cross} is that the flow passing through a vertex will not change and it is only the amount of flow passing through the edges that changes (Figure~\ref{fig:non-crossing}). This proves most of the constraints of LP~\ref{lp:lp} will still hold. The only two constraints that consider the flow passing through the edges, and not the vertices, are number ~\ref{cons:possitive} and number ~\ref{cons:consteq}. The former will be true since the process does not produce any negative flow and the latter is true since by Lemma~\ref{lem:crosscost} the total cost does not change.

Consequently, the following statement is true since we can first solve LP~\ref{lp:lp} by any polynomial time LP-solver and the run Algorithm~\ref{alg:cross} on its solution.
\begin{corollary}
	\label{lem:poly-lpanswer}
	There exists a polynomial time algorithm that finds $f\langle f_1, \dots, f_{\time}\rangle$, a collection of canonical flows, that is a solution of LP~\ref{lp:lp}, and for any  $t\in\timeset$  , $f_t$ does not contain any crossing flow.
\end{corollary}

\subsubsection{Constructing A Strategy}\label{sec:construction}

Assuming $f$ is a non-crossing solution of LP~\ref{lp:lp}, this section gives an algorithm to find a mixed strategy of the defender that equivalent to the set of canonical flows in $f$ and finally proves Theorem \ref{thm:main}. We first define what we mean by the top most flow path of a non-crossing canonical flow:

\begin{definition}[Top-Most Flow Path]

Let $f_t$ be a canonical flow of \graph{t} without any crossing flows. We say a canonical path $p = \langle e_1, e_2, \ldots, e_{\K+1}\rangle$ of \graph{t} is a flow path of $f_t$ if for any $k$ ($1\leq k \leq \K+1$), $\edgeflow{f_t}{e_k} > 0$. The top-most flow path of $f_t$ is the flow path of $f_t$ that is above all other flow paths of $f_t$ (that is well-defined because $f_t$ does not have any crossing flow). The size of the flow path $p$ of $f_t$, denoted by $\sizeofflowpath{p}$, is $m$ if for any $k$, $\edgeflow{f_t}{e_k} \geq m$ and there exists an edge $e_i$ of $p$ such that $\edgeflow{f_t}{e_i} = m$.
\end{definition}

\begin{lemma}\label{lem:topmost-compatible}
	Let $f = \langle f_1, \ldots, f_\time \rangle$ be a collection of non-crossing canonical flows that satisfies the compatibility constraint (constraint number \ref{cons:compatibility}) of LP~\ref{lp:lp}. For any $t$ that $\{t, t+1\}\subset \timeset$, the top-most flow paths of $f_t$ and $f_{t+1}$ are compatible. 
\end{lemma}

\begin{proof}
Let $\langle \source{t}, \vertex{t}{1}{y_1}, \dots, \vertex{t}{\K}{y_\K}, \sink{t} \rangle$ and $\langle \source{t+1}, \vertex{t+1}{1}{y'_1}, \dots, \vertex{t+1}{\K}{y'_\K}, \sink{t+1} \rangle$ respectively denote the top-most flow paths of $f_t$ and $f_{t+1}$. It suffices to prove for any $k\in \Kset$, there is a valid movement from the corresponding interval of $\vertex{t}{k}{y_k}$ to the corresponding interval of $\vertex{t+1}{k}{y'_k}$. To do so, we assume this is not the case and obtain a contradiction. Let $\vertex{t+1}{k}{y'_k}\notin \outfeas{t}{\vertex{t}{k}{y_k}}$ for some $k \in \Kset$, then one of the following conditions should hold:
\begin{enumerate}
	\item $\vertex{t+1}{k}{y'_k}$ is below the feasible range $\outfeas{t}{\vertex{t}{k}{y_k}}$.
	\item $\vertex{t+1}{k}{y'_k}$ is above the feasible range $\outfeas{t}{\vertex{t}{k}{y_k}}$.
\end{enumerate}

If the first condition is true, the contradiction is that $\vertex{t+1}{k}{y'_k}$ cannot be in the top-most flow path of $f_{t+1}$. To see this, note that we know by constraint~\ref{cons:compatibility} of LP~\ref{lp:lp} that the total flow passing through the vertices of $\outfeas{t}{\vertex{t}{k}{y_k}}$ is not less than the flow passing through $\vertex{t}{k}{y_k}$, therefore there is a vertex in $\outfeas{t}{\vertex{t}{k}{y_k}}$ (and above $\vertex{t+1}{k}{y'_k}$) with a non-negative flow and thus $\vertex{t+1}{k}{y'_k}$ cannot be in the top-most flow path of $f_{t+1}$.

If the second condition is true, the contradiction is that constraint~\ref{cons:compatibility} cannot be satisfied. To see this, note that since $\vertex{t}{k}{y_k}$ is in the top-most flow path of $f_t$, no flow passes through the vertices above it and therefore $\sum_{v \in \vertexcons{t}{k}{1}{y}}\outflow{f_t}{v} = 1$. However, since $\vertex{t+1}{k}{y'_k}$ is above the feasible range $\outfeas{t}{\vertex{t}{k}{y_k}}$, $\sum_{v' \in \outfeas{t}{\vertexcons{t}{k}{1}{y}}} \outflow{f_{t+1}}{v'} < 1$ which means constraint~\ref{cons:compatibility} that indicates the value of the latter summation should not be less than the former one, cannot be satisfied.
\end{proof}

\begin{theorem} \label{th:compatibility}
Let $f = \langle f_1, \ldots, f_\time \rangle$ be a solution of LP~\ref{lp:lp} without any crossing flows. There exists a polynomial time algorithm to find a mixed strategy of the defender that is equivalent to $f$. 
\end{theorem}

To prove Theorem~\ref{th:compatibility}, we show the following iterative algorithm constructs the desired mixed strategy in polynomial time:
\begin{enumerate}
	\item Find the top-most flow paths $p_1, \ldots, p_\time$ of $f_1, \ldots, f_\time$.
	\item Construct the pure strategy $p$, corresponding to $p_1, \ldots, p_\time$.
	\item Add $p$ to $s$ with probability $q = \min \sizeofflowpath{p_i}$.
	\item For any edge $e$ of any $p_i$, decrease $\edgeflow{f_i}{e}$ to $\edgeflow{f_i}{e} - q$.
	\item If there is any flow left in $f_1, \ldots, f_\time$, repeat all the steps.
\end{enumerate}
Note that by Lemma~\ref{lem:topmost-compatible}, if the compatibility constraint of LP~\ref{lp:lp} is satisfied, the top-most flow paths are compatible. Since at the first round of the algorithm, we have an actual solution of the LP, the compatibility constraint is obviously satisfied. To completely prove the correctness of this algorithm, we also need to show after each iteration, changing the flows does not violate the compatibility constraint. For space limitations, we left this part of the proof to the appendix. Furthermore, the running time of this algorithm is polynomial in the input size since in each iteration, the flow passing through at least one edge decreases to zero and the total number of edges is polynomial.

We are now ready to prove Theorem~\ref{thm:main}.

\begin{proof}[of Theorem~\ref{thm:main}]
Recall that by Lemma~\ref{lem:optinlp}, the optimal solution of LP~\ref{lp:lp} is a lower bound for the utility of the attacker when both players play their optimal (minimax) strategies. Also note that by Lemma~\ref{lem:poly-lpanswer} and Theorem \ref{th:compatibility}, we can construct a mixed strategy $s$ of the defender that is equivalent to the optimal solution of LP~\ref{lp:lp} in polynomial time. This means the maximum utility of the attacker when the defender plays $s$, is equal to its lower bound and therefore $s$ minimizes the maximum expected utility of the attacker: i.e., it is a minimax strategy of the defender.
\end{proof}

\section{Continuous Model}\label{sec:cont}
In this section we prove the following theorem for the continuous model:
\begin{theorem}\label{the:cont}
There exists a polynomial time algorithm to find an optimal solution for \continuousProblem{}.
\end{theorem}

The given proof is based on a technical assumption that all numbers in the input are rational.

\begin{proof}[of Theorem~\ref{the:cont}]
The main idea of this proof is to reduce any instance of \continuousProblem{} to an instance of \discreteProblem{}, for which we know there exists a polynomial time algorithm.

Recall that any rational number can be represented by a fraction $\frac{a}{b}$ such that both $a$ and $b$ are integers. Let $B$ be the set of denominators in this fractional representation of all numbers in the input. We define $m$ to be the product of all numbers in $B$. Note that the number of digits needed to represent $m$ is polynomial in the input size since every number in $B$ appears in the input. To create an instance of \discreteProblem{}, we multiply all target positions, $\D$, $\R$ and $\M$, given in the instance of \continuousProblem{} to $m$.

To use the algorithm for \discreteProblem{}, it suffices to prove in the scaled solutions, there exists an optimal solution that places patrols only in the integer locations. To do this, we prove that for any given patrol path $p_1=\langle m_i \rangle_\time$ in the scaled version, there exists a patrol path $p_2=\langle m'_i \rangle_\time$ that covers the same set of targets and for any $t\in\timeset$, $m'_t$ is an integer position. It suffices to set $m'_t$ to be $\lfloor m_t \rfloor$. Note that since the position of all targets and the protecting ranges of the patrols in the scaled version are all integers, this patrol path protects exactly the same set of targets. Now we can use the algorithm of Theorem~\ref{thm:main} to find the optimal solution of this scaled input and then scale it back to the original size by dividing the patrols' locations by $m$.
\end{proof}

\bibliographystyle{plain}
\bibliography{ref}

\begin{thebibliography}{10}

\bibitem{ahmadinejad2016duels}
AmirMahdi Ahmadinejad, Sina Dehghani, MohammadTaghi Hajiaghayi, Brendan Lucier,
  Hamid Mahini, and Saeed Seddighin.
\newblock From duels to battefields: Computing equilibria of blotto and other
  games.
\newblock In {\em Proceedings of the Thirtieth {AAAI} Conference on Artificial
  Intelligence}, 2016.

\bibitem{Balcan-ec15}
Maria{-}Florina Balcan, Avrim Blum, Nika Haghtalab, and Ariel~D. Procaccia.
\newblock Commitment without regrets: Online learning in stackelberg security
  games.
\newblock In {\em Proceedings of the Sixteenth {ACM} Conference on Economics
  and Computation, {EC}}, pages 61--78, 2015.

\bibitem{behnezhad2016faster}
Soheil Behnezhad, Sina Dehghani, Mahsa Derakhshan, MohammadTaghi HajiAghayi,
  and Saeed Seddighin.
\newblock Faster and simpler algorithm for optimal strategies of blotto game.
\newblock In {\em Proceedings of the Thirty-First {AAAI} Conference on
  Artificial Intelligence}, 2017.

\bibitem{Blum-nips14}
Avrim Blum, Nika Haghtalab, and Ariel~D. Procaccia.
\newblock Learning optimal commitment to overcome insecurity.
\newblock In {\em Annual Conference on Neural Information Processing Systems
  (NIPS)}, pages 1826--1834, 2014.

\bibitem{bovsansky2011computing}
Branislav Bo{\v{s}}ansk{\`y}, Viliam Lis{\`y}, Michal Jakob, and Michal
  P{\v{e}}chou{\v{c}}ek.
\newblock Computing time-dependent policies for patrolling games with mobile
  targets.
\newblock In {\em The 10th International Conference on Autonomous Agents and
  Multiagent Systems-Volume 3}, pages 989--996. International Foundation for
  Autonomous Agents and Multiagent Systems, 2011.

\bibitem{OneSize-aaai16}
Matthew Brown, Arunesh Sinha, Aaron Schlenker, and Milind Tambe.
\newblock One size does not fit all: {A} game-theoretic approach for
  dynamically and effectively screening for threats.
\newblock In {\em Proceedings of the Thirtieth {AAAI} Conference on Artificial
  Intelligence}, pages 425--431, 2016.

\bibitem{Conitzer-ec06}
Vincent Conitzer and Tuomas Sandholm.
\newblock Computing the optimal strategy to commit to.
\newblock In {\em Proceedings 7th {ACM} Conference on Electronic Commerce
  (EC-2006)}, pages 82--90, 2006.

\bibitem{fang2013optimal}
Fei Fang, Albert~Xin Jiang, and Milind Tambe.
\newblock Optimal patrol strategy for protecting moving targets with multiple
  mobile resources.
\newblock In {\em Proceedings of the 2013 international conference on
  Autonomous agents and multi-agent systems}, pages 957--964. International
  Foundation for Autonomous Agents and Multiagent Systems, 2013.

\bibitem{PAWS-aaai16}
Fei Fang, Thanh~Hong Nguyen, Rob Pickles, Wai~Y. Lam, Gopalasamy~R. Clements,
  Bo~An, Amandeep Singh, Milind Tambe, and Andrew Lemieux.
\newblock Deploying {PAWS:} field optimization of the protection assistant for
  wildlife security.
\newblock In {\em Proceedings of the Thirtieth {AAAI} Conference on Artificial
  Intelligence}, pages 3966--3973, 2016.

\bibitem{kiekintveld2009computing}
Christopher Kiekintveld, Manish Jain, Jason Tsai, James Pita, Fernando
  Ord{\'o}{\~n}ez, and Milind Tambe.
\newblock Computing optimal randomized resource allocations for massive
  security games.
\newblock In {\em Proceedings of The 8th International Conference on Autonomous
  Agents and Multiagent Systems-Volume 1}, pages 689--696. International
  Foundation for Autonomous Agents and Multiagent Systems, 2009.

\bibitem{korzhyk2010complexity}
Dmytro Korzhyk, Vincent Conitzer, and Ronald Parr.
\newblock Complexity of computing optimal stackelberg strategies in security
  resource allocation games.
\newblock In {\em AAAI}, 2010.

\bibitem{Conitzer-aaai13}
Joshua Letchford and Vincent Conitzer.
\newblock Solving security games on graphs via marginal probabilities.
\newblock In {\em Proceedings of the Twenty-Seventh {AAAI} Conference on
  Artificial Intelligence}, 2013.

\bibitem{Letchford-sagt09}
Joshua Letchford, Vincent Conitzer, and Kamesh Munagala.
\newblock Learning and approximating the optimal strategy to commit to.
\newblock In {\em Algorithmic Game Theory, Second International Symposium,
  {SAGT}}, pages 250--262, 2009.

\bibitem{Marecki-aamas12}
Janusz Marecki, Gerald Tesauro, and Richard Segal.
\newblock Playing repeated stackelberg games with unknown opponents.
\newblock In {\em International Conference on Autonomous Agents and Multiagent
  Systems, {AAMAS}}, pages 821--828, 2012.

\bibitem{SSG-ML-survey16}
Giuseppe~De Nittis and Francesco Trov{\`{o}}.
\newblock Machine learning techniques for stackelberg security games: a survey.
\newblock {\em Technical report on {\tt arxiv.org}}, abs/1609.09341, 2016.

\bibitem{pita2008deployed}
James Pita, Manish Jain, Janusz Marecki, Fernando Ord{\'o}{\~n}ez, Christopher
  Portway, Milind Tambe, Craig Western, Praveen Paruchuri, and Sarit Kraus.
\newblock Deployed armor protection: the application of a game theoretic model
  for security at the los angeles international airport.
\newblock In {\em Proceedings of the 7th international joint conference on
  Autonomous agents and multiagent systems: industrial track}, pages 125--132.
  International Foundation for Autonomous Agents and Multiagent Systems, 2008.

\bibitem{Tim-EQ-survey10}
Tim Roughgarden.
\newblock Computing equilibria: A computational complexity perspective.
\newblock {\em Economic Theory}, 42(1):193--236, 2010.

\bibitem{tambe2011security}
Milind Tambe.
\newblock {\em Security and game theory: algorithms, deployed systems, lessons
  learned}.
\newblock Cambridge University Press, 2011.

\bibitem{Xu-ec16}
Haifeng Xu.
\newblock The mysteries of security games: Equilibrium computation becomes
  combinatorial algorithm design.
\newblock In {\em Proceedings of the 2016 {ACM} Conference on Economics and
  Computation, {EC}}, pages 497--514, 2016.

\bibitem{xu2014solving}
Haifeng Xu, Fei Fang, Albert~Xin Jiang, Vincent Conitzer, Shaddin Dughmi, and
  Milind Tambe.
\newblock Solving zero-sum security games in discretized spatio-temporal
  domains.
\newblock In {\em AAAI}, pages 1500--1506. Citeseer, 2014.

\bibitem{yin2012trusts}
Zhengyu Yin, Albert~Xin Jiang, Milind Tambe, Christopher Kiekintveld, Kevin
  Leyton-Brown, Tuomas Sandholm, and John~P Sullivan.
\newblock Trusts: Scheduling randomized patrols for fare inspection in transit
  systems using game theory.
\newblock {\em AI Magazine}, 33(4):59, 2012.

\end{thebibliography}

\clearpage

\appendix

\section{Missing Proofs}

\subsection*{Lemma~\ref{lem:sameintervalsameprotection}}
\xhdr{Statement.}
 Let $k$ and $k'$ be two patrols in the same interval at any time $t$. The set of targets that $k$ and $k'$ protect at time $t$ are equal.

\begin{proof}
	We suppose this is not the case and obtain a contradiction. Without losing generality assume $k$ protects a target $a$ at time $t$ that $k'$ does not. Lines \ref{line:R} and \ref{line:R+1} of Algorithm \ref{alg:partition} indicate there are two interval points $p_i$ and $p_j$ at $\targetpos{a}{t}-\R$ and $\targetpos{a}{t}+\R+\epsdistance$ respectively. Assume $\epsilon$ is too small that there is no valid patrol position in the non-inclusive range between $\targetpos{a}{t}+\R$ and $\targetpos{a}{t}+\R+\epsdistance$. This means a patrol has a distance of at most $\R$ from $a$ (or simply protects $a$) if and only if it is in an interval between $p_i$ and $p_j$.  Note that we assumed $k'$ does not protect $a$, and hence it is not between $p_i$ and $p_j$, while $k$ has to be between them to protect $a$. This means $k$ and $k'$ could not be in the same interval, which is a contradiction.
\end{proof}

\subsection*{Lemma~\ref{lem:movinginintervals}}
\xhdr{Statement.} Let $[s_i, f_i)$ and $[s_j, f_j)$ be two arbitrary intervals in $\intseq{t}$ and $\intseq{t+1}$ respectively. If there exists a feasible move from an arbitrary position in $[s_i, f_i)$ to a position in $[s_j, f_j)$, for any position in $[s_i, f_i)$, there exists a feasible move to a position in $[s_j, f_j)$.

\begin{proof}
Suppose there exists a feasible move from a position $x_i$ in interval $[s_i, f_i)$ to a position $x_j$ in interval $[s_j, f_j)$.  The existence of a feasible move from $x_i$ to $[s_j, f_j)$ implies $s_j-\D \leq x_i< f_j+\D$. Also note that Line \ref{line:+D} and Line \ref{line:-D} of Algorithm \ref{alg:partition} indicate $f_j+\D$ and $s_j-\D$ are in \parof{t}. Therefore, since $[s_i, f_i)$ is the interval containing $x_i$, the following equation holds:
$$s_j-\D \leq s_i\leq x_i<f_i\leq f_j+\D$$	
This means any possible location $x'_i$ in $[s_i, f_i)$ satisfies $s_j-\D \leq x'_i< f_j+\D$, and therefore has a feasible move to $[s_j, f_j)$.
\end{proof}

\subsection*{Lemma~\ref{lem:surjective}}
\xhdr{Statement.}
For any interval path $\xi = \langle \intseqitem{i}{x_i}\rangle_\time$, there is at least one patrol path in $S(\xi)$.

\begin{proof}
By Definition~\ref{def:intervalpath}, for any time point $t \in \timeset$, there is at least one feasible move from a position in $\intseqitem{t}{x_{t}}$ to a position in $\intseqitem{t+1}{x_{t+1}}$, also based on Lemma \ref{lem:movinginintervals}, existence of a feasible move from interval $\intseqitem{t}{x_{t}}$ to interval $\intseqitem{t+1}{x_{t+1}}$ means there is a feasible move from any position in $\intseqitem{t}{x_{t}}$ to at least one position in $\intseqitem{t+1}{x_{t+1}}$. Therefore any patrol starting in any position in $\intseqitem{1}{x_1}$, could reach at least one position in $\intseqitem{\time}{x_{\time}}$ by feasible moves, which forms a patrol path in $S(\xi)$. 
\end{proof}

\subsection*{Lemma~\ref{lem:purelemma}}
\xhdr{Statement.} Let $s$ be a pure snapshot at time $t$, and $a$ be an arbitrary target. The payoff of the attacker with respect to $s$, if he attacks the target $a$ at time $t$, equals the cost of the target $a$ in the canonical path equivalent (Definition~\ref{def:puremapping}) to $s$.

\begin{proof}
Let $p = \langle \source{t}, \vertex{t}{1}{y_1}, \vertex{t}{2}{y_2}, \ldots, \vertex{t}{\K}{y_{\K}}, \sink{t} \rangle$ be the canonical path equivalent to $s$ and $E=\{ e_1, e_2, \ldots, e_{\K+1} \}$ be the set of edges in this canonical path. By Definition~\ref{def:puremapping}, $s$ contains \K{} patrols in intervals $\intseqitem{t}{y_1}, \intseqitem{t}{y_2}, \ldots, \intseqitem{t}{y_{\K}}$, which are already sorted based on position. 
 	Let $[s_i, f_i)$ denote the start and end positions of interval \intseqitem{t}{y_i}. If target $a$ is not covered by $s$ exactly one of the following conditions holds:
\begin{enumerate}
	\item There exists a single $i$ such that $1\leq i \leq \K -1$ and $f_i \leq \targetpos{t}{a} < s_{i+1}$
	\item $\targetpos{t}{a} < s_{1}$
	\item $f_{\K} \leq \targetpos{t}{a}$
\end{enumerate}
 
This indicates that if target $a$ is not covered by $s$, there exists exactly one $i$ such that  $1\leq i\leq \K+1$  and $\edgecover{e_i}{a}=1$ (See Definition \ref{def:canonical}), but if it is covered $\edgecover{e_i}{a}=0$ for all the edges in this path. By Definition \ref{def:puremapping}, the cost of the target $a$ in the canonical path equivalent to $s$ is defined as follows:
	 $$\Sigma_{i=1}^{\K+1} \edgecover{e_i}{a}\times \targetweight{t}{a}$$  This formula equals to $0$ if the target is covered, and it equals to $\targetweight{t}{a}$ otherwise, which is equal to the payoff of the attacker with respect to $s$, if he attacks target $a$ at time $t$.

\end{proof}

\subsection*{Lemma~\ref{lem:mixedcost}}
\xhdr{Statement.}
Let $r$ be a mixed snapshot at time $t$ and let $f$ denote the canonical flow that $r$ is mapped to. The maximum expected payoff of the attacker at time $t$ with respect to $r$, equals the cost of $f$.

\begin{proof}
	Let $\{(d_i, p_i)\}_{n}$ denote the mixed snapshot $r$. Lemma \ref{lem:purelemma} states that the payoff of the attacker attacking target $a$ while the pure snapshot $d_i$ represents the placement of patrols at time $t$ is equal to: \begin{equation*}
		 \Sigma_{e \in E_i} \edgecover{e}{a}\times \targetweight{t}{a} 	\end{equation*}
		where $E_i$ denotes the set of edges in the canonical path equivalent to pure snapshot $d_i$. This indicates that the expected payoff of the attacker attacking target $a$ with respect to $r$ is equal to:
		
		\begin{equation*}
		\Sigma^{n}_{i=1} p_i \times (\Sigma_{e \in E_i} \edgecover{e}{a}\times \targetweight{t}{a}) = \Sigma_{e \in \edgeset{\graph{t}}} \edgeflow{f}{e}\times\edgecover{e}{a}\times \targetweight{t}{a} 	\end{equation*}	
So, the maximum expected payoff of the attacker while mixed snapshot $r$ represents the placement of patrols at time $t$ equals to the cost of the canonical flow $f$ that is defined in Definition \ref{def:canonical} as follows:
\begin{equation*}
		\max \Sigma_{e \in \edgeset{\graph{t}}} \edgeflow{f}{e}\times\edgecover{e}{a}\times \targetweight{t}{a} \qquad \forall a; \text{ where } a \in \targetset.
	\end{equation*}
	
\end{proof}

\subsection*{Lemma~\ref{lemma: noncrossingintervalpath}} 
\xhdr{Statement.}
If $\langle \intseqitem{t}{y^{t}_i}\rangle_\K{}$ and $\langle \intseqitem{t+1}{y^{t+1}_i}\rangle_\K{}$ are two pure snapshots at time $t$ and $t+1$ in at least one valid pure strategy $p$, then for any $j\in \Kset$ we have $\intseqitem{t+1}{y^{t+1}_j} \in \outfeas{t}{\intseqitem{t}{y^{t}_j}}$.

\begin{proof}
	We prove this lemma by induction on $\K$.
	
	Induction hypothesis: we assume this lemma holds for any $\K$ where $\K<n$. 
	For $\K=1$ , the base case, each snapshot has one patrol and they have to be compatible.
	
	Induction step: We show that if $\langle \intseqitem{t}{y^{t}_i} \rangle_n$ and $\langle \intseqitem{t+1}{y^{t+1}_i} \rangle_n$ respectively denote a pure snapshot at time $t$ and $t+1$  in at least one pure strategy $p_n$,  $\intseqitem{t+1}{y_j} \in \outfeas{t}{\intseqitem{t}{y_j}}$ for any $j$ where $j\leq n$. Let  $\langle \intseqitem{1}{z_1}, \dots, \intseqitem{\time}{z_{\time}} \rangle$ and $\langle \intseqitem{1}{w_1}, \dots, \intseqitem{\time}{w_{\time}} \rangle$ denote two interval paths in pure strategy $p_n$ such that the following equations hold:
	\begin{enumerate}
	\item $\intseqitem{t}{z_t}=\intseqitem{t}{y^{t}_n}$
	\item $\intseqitem{t+1}{z_{t+1}}=\intseqitem{t+1}{y^{t+1}_l}$
	\item $\intseqitem{t}{w_t}=\intseqitem{t}{y^{t}_m}$
	\item $\intseqitem{z+1}{w_{t+1}}=\intseqitem{t+1}{y^{t+1}_n}$
	\end{enumerate}
We first prove that  $\intseqitem{t+1}{y^{t+1}_n} \in \outfeas{t}{\intseqitem{t}{y^{t}_n}} $ and  $\intseqitem{t+1}{y^{t+1}_l} \in \outfeas{t}{\intseqitem{t}{y^{t}_m}}$.

		To prove $\intseqitem{t+1}{y^{t+1}_n} \in \outfeas{t}{\intseqitem{t}{y^{t}_n}} $  we first assume that this is not the case, then we obtain a contradiction. Let $s^{t}_i$ and $t^{t}_i$ respectively denote the start and end points of the interval  $\intseqitem{t}{y^{t}_i}$, and let $s^{t+1}_i$ and $t^{t+1}_i$ denote the starts and end points of the interval   $\intseqitem{t+1}{y^{t+1}_i}$ for all $i$ such that $1\leq i \leq n$. if $\intseqitem{t+1}{y^{t+1}_n} \notin \outfeas{t}{\intseqitem{t}{y^{t}_n}} $ exactly one of the following conditions holds:
	\begin{enumerate}
		\item $t^{t}_n+\R \leq s^{t+1}_n$: Since $m \leq n$, in this case the inequality $t^{t}_m+\R \leq t^{t}_n+\R \leq s^{t+1}_n$ holds, which indicates  $\intseqitem{t+1}{y^{t+1}_n} \notin \outfeas{t}{\intseqitem{t}{y^{t}_m}}$. This is a contradiction with the existence of the interval path $\langle \intseqitem{1}{w_1}, \dots, \intseqitem{t}{y^{t}_m}, \intseqitem{t+1}{y^{t+1}_n} , \dots, \intseqitem{\time}{w_{\time}} \rangle$ in the pure strategy $p_n$.
		\item $t^{t+1}_n+\R \leq s^{t}_n$: Since $l \leq n$, the inequality $t^{t+1}_l+\R \leq t^{t+1}_n+\R \leq s^{t}_n$ holds in this case. This means $\intseqitem{t+1}{y^{t+1}_l} \notin \outfeas{t}{\intseqitem{t}{y^{t}_n}}$, which is a contradiction with the existence of the interval path $\langle \intseqitem{1}{z_1}, \dots, \intseqitem{t}{y^{t}_n}, \intseqitem{t+1}{y^{t+1}_l}, \dots, \intseqitem{\time}{z_{\time}} \rangle$ in the pure strategy $p_n$. 
	\end{enumerate}
		We proved $\intseqitem{t+1}{y^{t+1}_n} \in \outfeas{t}{\intseqitem{t}{y^{t}_n}} $. It is easy to see that  $\intseqitem{t+1}{y^{t+1}_l} \in \outfeas{t}{\intseqitem{t}{y^{t}_m}}$ is also correct in the same way.

	 Since  $\intseqitem{t+1}{y^{t+1}_m} \in \outfeas{t}{\intseqitem{t}{y^{t}_l}}$, $\langle \intseqitem{1}{w_1}, \dots, \intseqitem{t}{w_t}, \intseqitem{t+1}{z_{t+1}} , \dots, \intseqitem{\time}{z_{\time}} \rangle$ is a valid interval path. Using this fact, we construct the pure strategy $p_{n-1}$ by deleting interval paths $\langle \intseqitem{1}{z_1}, \dots, \intseqitem{\time}{z_{\time}} \rangle$ and $\langle \intseqitem{1}{w_1}, \dots , \intseqitem{\time}{w_{\time}} \rangle$ from $p_n$, and adding the interval path $\langle \intseqitem{1}{w_1}, \dots, \intseqitem{t}{w_t}, \intseqitem{t+1}{z_{t+1}} , \dots, \allowbreak \intseqitem{\time}{z_{\time}} \rangle$ to it. One can easily see that $\langle \intseqitem{t}{y^{t}_i} \rangle_{n-1}$ and $\langle \intseqitem{t+1}{y^{t+1}_i} \rangle_{n-1}$ respectively are the pure snapshots at time $t$ and $t+1$, in pure strategy $p_{n-1}$. Hence in this case $\K=n-1$, by induction hypothesis, $\intseqitem{t+1}{y^{t+1}_j} \in \outfeas{t}{\intseqitem{t}{y^{t}_j}}$ for any $j$ such that $1\leq j\leq n-1$. We also proved that $\intseqitem{t+1}{y^{t+1}_n} \in \outfeas{t}{\intseqitem{t}{y^{t}_n}} $. Therefore for any $j$ where $1\leq j\leq n$,  $\intseqitem{t+1}{y^{t+1}_j} \in \outfeas{t}{\intseqitem{t}{y^{t}_j}}$,  and the proof of the induction step is complete.

\end{proof}

\subsection*{Lemma~\ref{lem: sortingintervalpath}}
\xhdr{Statement.}
	There exists an optimal mixed strategy of the defender, such that  for every pure strategy $p$ in its support there is an ordering of interval paths $\langle\xi_1, \dots, \xi_\K  \rangle$ such that the following condition holds for this ordering:
	for any two interval paths $\xi_i$ and $\xi_j$, in the pure strategy p,  $\xi_i$ is intervally under $\xi_j$ if $i\leq j$.
	
\begin{proof}
	We first prove that for any pure strategy $p_1$ there exists a pure strategy $p_2$ that protects the same set of targets as $p_1$, and there is an ordering of interval paths $\langle\xi_1, \dots, \xi_\K  \rangle$ such that interval paths $\xi_i$ is intervally under $\xi_j$ if $i\leq j \leq\K$.
		
	Let pure snapshot $\langle \intseqitem{t}{y^{t}_1}, \intseqitem{t}{y^{t}_2}, \dots, \intseqitem{t}{y^{t}_\K} \rangle_\K$ denote the patrols' placement in pure strategy $p_1$ at time $t$. We construct the pure strategy $p_2$, such that even though it might have different set of interval paths from $p_1$, they share the same set of pure snapshots. Let $\Phi$ denote the set of interval paths in $p_2$. At first $\Phi={\varnothing}$. Then, for any patrol $k \in \Kset$ we add interval path 
	$\xi_k=\langle \intseqitem{1}{y^{1}_k},  \intseqitem{2}{y^{2}_k}, \dots, \intseqitem{\time}{y^{\time}_k} \rangle$ to $\Phi$ (The Interval path $\xi_k$ contains the $k$-th interval of all the pure snapshots of $p_1$). Since in Lemma \ref{lemma: noncrossingintervalpath} we proved that $\intseqitem{t+1}{y^{t+1}_k} \in \outfeas{t}{\intseqitem{t}{y^{t}_k}}$, the set $\Phi$ consists of $\K$ valid interval paths. Also for any $\xi_i$ and $\xi_j$ in  $\Phi$ interval path $\xi_i$ is intervally under the interval path $\xi_j$ if $1 \leq i\leq j \leq \K$. Since $p_1$ and $p_2$ share the same set of pure snapshots, they also protect the same set of targets. 
	
	To prove this lemma, let $m$ denote an optimal mixed strategy of the defender. We replace any pure strategy in the support of $m$ with its modified version
So, the utility of defender playing this mixed strategy does not change, and it is still an optimal strategy. Also, for any $p$ in the support of $m$ 
	if we order the interval paths from $\xi_1$ to $\xi_\K$, the following condition holds: for any two interval path $\xi_i$ and $\xi_j$, in the pure strategy p,  $\xi_i$ is intervally under $\xi_j$ if $1\leq i\leq j \leq\K$
\end{proof}

\subsection*{Lemma~\ref{lem:optinlp}}
\xhdr{Statement.}
The solution of Linear Program~\ref{lp:lp} gives a lower bound for the utility of the attacker when both players play their optimal strategies.

\begin{proof}
	Based on Lemma~\ref{lem: sortingintervalpath}, there exists an optimal mixed strategy $s$ such that for any pure strategy $p$ in support of $s$ this condition holds: there exists an ordered set $\langle \xi_1, \xi_2, \dots, \xi_\K \rangle$ of interval paths in $p$  such that for any $i,j$ that $1\leq i\leq j \leq \K$, interval path $\xi_i$ is intervally under $\xi_j$.
	Let $u_{opt}$ denote the attacker's utility in mixed strategy $s$, and let mixed snapshot $m_t$ denote the placement of patrols at time $t$ in mixed strategy $s$. Also $f^s_t$ denotes the corresponding canonical flow of mixed snapshot $m_t$. We set values of the variables in LP~\ref{lp:lp} as follows:
	\begin{enumerate}
		\item $u=u_{opt}$
		\item $f_t=f^s_t \quad \forall t\in \timeset$
	\end{enumerate}
	
	Then, we prove that under this assignments all the constraints of the LP are satisfied. The set of constraints in lines \ref{cons:start} to \ref{cons:possitive} of  this LP are the necessary conditions for a canonical flow, so for any $t\in \timeset$, canonical flow $f^s_t$ satisfies them. There also exists another set of constraints in line \ref{cons:compatibility}, that is necessary for the compatibility of two consecutive canonical flows. Let $<\xi^p_1, \xi^p_2, \dots, \xi^p_\K>$ denote interval paths in pure strategy $p$ in support of $s$, such that $\xi^p_b$ is intervally under $\xi^p_a$ for any $a$ and $b$ that $1\leq a\leq b\leq \K$.
	Note that, if at time $t$ patrol $k_a$ is in an interval in the set of intervals $\intseqcons{t}{i}{j}$ such that $1\leq i\leq \parsize{t}$ and $\{t, t+1\}\subset \timeset$, $k_a$ is in an interval in the set $\outfeas{t}{\intseqcons{t}{i}{j}}$ at time $t+1$. This indicates that if at time $t$, with probability $p_a$ patrol $k_i$ is in an interval in the set $\intseqcons{t}{i}{j}$, with probability at least $p_a$ at time $t+1$, $k_i$ is in an interval in the set $\outfeas{t}{\intseqcons{t}{i}{j}}$. In LP \ref{lp:lp} , $\graph{t}$ denotes the day graph corresponding to the mixed snapshot $t$ in strategy $s$, and for any $k\in \Kset$, the amount of flow crossing through the set of vertices \vertexcons{t}{k}{i}{j} in $\graph{t}$ denotes the probability that at time $t$, patrol $k$ is in $\intseqcons{t}{i}{j}$. So the following equation holds:
	
	$$\Sigma_{v \in \vertexcons{t}{k}{i}{j}}\edgeflow{f^{s+}_t}{v} \leq \Sigma_{v' \in \outfeas{t}{\vertexcons{t}{k}{i}{j}}} \edgeflow{f^{s+}_{t+1}}{v'}  \forall t, k, i, j: t, t+1 \in \timeset, k\in \Kset, 1 \leq i \leq j \leq \parsize{t}$$
	
		 In addition, the set of constraints in the line \ref{cons:consteq} of the LP is also satisfied since at least one of them is violated only if the expected payoff of the attacker, attacking an arbitrary target $a$ is more than $u_{opt}$, but $u_{opt}$ is the maximum payoff of the attacker while defender plays strategy $s$.
		 
		 We proved that there exists a valid assignment to variables of this LP such that $u=u_{opt}$, so $u_{opt}$ is an upper bound for the value of $u$ in this LP.
\end{proof}

\subsection*{Lemma~\ref{lem:crosscost}} 
\xhdr{Statement.}
	Function \textproc{ResolveCrosses} of Algorithm \ref{alg:cross} does not increase the total cost of the input canonical flow.

\begin{proof}As mentioned in Definition \ref{def:canonical} the cost of a canonical flow is as follows:
	\begin{equation*}
		\max \Sigma_{e \in \edgeset{\graph{t}}} \edgeflow{f}{e}\times\edgecover{e}{a}\times \targetweight{t}{a} \qquad \forall a; \text{ where } a \in \targetset{t}.
	\end{equation*}
	The only thing in this function that can affect the above mentioned cost is that the amount of flow crossing through both edges  $e_1=(\vertex{t}{x}{y_1}, \vertex{t}{x+1}{y_2})$ and $e_2=(\vertex{t}{x}{y'_1},\vertex{t}{x+1}{y'_2})$ decreases by a particular amount, which is then added to the flow crossing through edges $e_3=(\vertex{t}{x}{y_1}, \vertex{t}{x+1}{y'_2})$ and $e_4=(\vertex{t}{x}{y'_1},\vertex{t}{x+1}{y_2})$. When the amount of flow passing through $e_1$ and $e_2$ decreases by $f_m$ the cost of choosing an arbitrary target $a$  decreases by $
		f_m\times \targetweight{t}{a}(\edgecover{e_1}{a}+ \edgecover{e_2}{a})$, and when this amount is added to flow of $e_3$ and $e_4$, the mentioned cost increases by $f_m\times\targetweight{t}{a} (\edgecover{e_3}{a}+ \edgecover{e_4}{a})$. 
		Since we want to show that the maximum cost for all the targets does not increase, it suffices to prove that the amount of increment in cost for each target is less than the decrement. In other words we prove the following relation holds for any arbitrary target: 
		\begin{equation} \label{eq:crossing}
		(\edgecover{e_3}{a}+ \edgecover{e_4}{a}) \leq 
		(\edgecover{e_1}{a}+ \edgecover{e_2}{a})
		\end{equation}
		Since $e_1$ and $e_2$ are crossing edges they have one of the conditiones mentioned in Definition \ref{def:crossingedges}.   Without loss of generality we assume the first condition holds, which means $y_1<y'_1$ and $y_2>y'_2$. Moreover, by Definition~\ref{def:graph}, we have $y_1 \leq y_2$ and $y'_1 \leq y'_2$, which yields $y_1< y'_1 \leq y'_2< y_2$. Let $e=(\vertex{t}{x}{y}, \vertex{t}{x+1}{y'})$ denote an arbitrary edge in \graph{t}, where $\intseqitem{t}{y}$ and $\intseqitem{t}{y'}$ are intervals corresponding to vertices $\vertex{t}{x}{y}$ and $\vertex{t}{x+1}{y'}$. Recall that by Definition \ref{def:graph} for any target $a$ , $\edgecover{e}{a}$ is equal to $1$ if the two following conditions hold: (1) target $a$ is located in a position between $\intseqitem{t}{y}$ and $\intseqitem{t}{y'}$ (non-inclusive) at time $t$, and, (2) target $a$ could not be protected by any patrol at any arbitrary position in $\intseqitem{t}{y}$ or $\intseqitem{t}{y'}$; otherwise  $\edgecover{e}{a}$ is equal to 0.  (Here, we ignore those edges that one of their connected vertices is \sink{t} or \source{t} because these edges do not cross)
		
	In the following, we prove that equation ~\ref{eq:crossing} holds for any possible value of (\edgecover{e_3}{a}+ \edgecover{e_4}{a}).
		
		\begin{itemize}
		\item \edgecover{e_3}{a}+ \edgecover{e_4}{a}= 0: Since the right side of the equation \ref{eq:crossing} is not less than $0$ the inequality holds in this case.
	
		\item $\edgecover{e_3}{a}+ \edgecover{e_4}{a}= 2$: In this case, $\edgecover{e_3}{a}= 1$ and $\edgecover{e_4}{a}= 1$. So, target $a$ is located in a position between	 intervals $\intseqitem{t}{y'_1}$ and $\intseqitem{t}{y_2}$, and between intervals  $\intseqitem{t}{y_1}$ and $\intseqitem{t}{y'_2}$. Moreover, target $a$ can not  be protected by any patrol that is either in interval $\intseqitem{t}{y_1}$, $\intseqitem{t}{y'_1}$, $\intseqitem{t}{y_2}$ or $\intseqitem{t}{y'_2}$. This indicates that $\edgecover{e_1}{a}= 1$ and $\edgecover{e_2}{a}= 1$ since $\targetpos{t}{a}$, the position of target $a$ at time $t$, is above the intervals $y_1$ and $y'_1$ and below the intervals  $y'_2$ and $y_2$ , which indicates that the target is in a position between $y_1$ and $y_2$, and between  $y'_1$ and $y'_2$. So, in this case $\edgecover{e_1}{a}+ \edgecover{e_2}{a}= 2$, and equality \ref{eq:crossing} holds.
			
		\item $\edgecover{e_3}{a}+ \edgecover{e_4}{a}= 1$: In this case, $\edgecover{e_3}{a}= 1$ or $\edgecover{e_4}{a}= 1$. So, target $a$ is located in a position between	 intervals $\intseqitem{t}{y'_1}$ and $\intseqitem{t}{y_2}$ or it is in a position between intervals  $\intseqitem{t}{y_1}$ and $\intseqitem{t}{y'_2}$.  This indicates that $\edgecover{e_1}{a}=1$ since $\targetpos{t}{a}$, the position of target $a$ at time $t$, is above the interval $y_1$ and and below the interval $y_2$. Moreover, one can easily see that having $\edgecover{e_3}{a}= 1$ or $\edgecover{e_4}{a}= 1$ yields that target $a$ can not  be protected by any patrol that is either in interval $\intseqitem{t}{y_1}$ or $\intseqitem{t}{y_2}$, so $1\leq \edgecover{e_1}{a}+ \edgecover{e_2}{a}$ holds.
		\end{itemize}
		
The three mentioned cases for the value of \edgecover{e_3}{a}+ \edgecover{e_4}{a}  cover all possible cases, so equality \ref{eq:crossing} holds and this lemma is proved.
\end{proof}

\subsection*{Lemma~\ref{lem:crosstime}} 
\xhdr{Statement.}
The running time of the function \textproc{ResolveCrosses} in the Algorithm \ref{alg:cross} is polynomial.

\begin{proof}The main idea of this proof is that after each call of the function \textproc{ResolveCross}, the minimum crossing flow gets greater. (two crossing flows are compared based on the comparison defined in the Definition \ref{def:comparison}.) Since there are polynomially many pair of crossing edges the run time of this algorithm is polynomial as well. 
	
	Two edges forming the minimum cross flow are $e_1=(\vertex{t}{x}{y_1}, \vertex{t}{x+1}{y_2})$ and $e_2=(\vertex{t}{x}{y'_1},\vertex{t}{x+1}{y'_2})$. Since $e_1$ and $e_2$ are crossing edges  one of the conditions mentioned in Definition \ref{def:crossingedges} holds for them. Without loos of generality we assume the first condition holds, which means $y_1<y'_1$ and $y_2>y'_2$. By Definition \ref{def:comparison} and \ref{def:crossingedges} it is not possible for the flow passing through the edge $e_5=(\vertex{t}{x}{y"_1},\vertex{t}{x+1}{y"_2})$ to be greater than zero if it has one of the following conditions:
	\begin{itemize}
	\item $y_1 < y"_1$ and $y"_2 \leq y'_2$: The assumptions $y"_2\leq y'_2$ and $y_2>y'_2$ result that $y"_2<y_2$. Since  $y"_2<y_2$ and $y_1 < y"_1$, by  Definition \ref{def:crossingedges} $e_1$ and $e_5$ cross which contradicts with the fact that the pair of crossing edges $(e_1, e_2)$ are the minimum crossing flow, hence by  Definition \ref{def:comparison} crossing pair $(e_1, e_5)$ is less than the pair $(e_1, e_2)$.
	
	\item $y"_1\leq y_1$ and $y'_2< y"_2$: In this case, since $y_1<y'_1$ and  $y"_1\leq y_1$ hold, the inequality $y"_1<y'_1$ also holds. In addition, hence  $y"_1<y'_1$ and $y'_2< y"_2$, by Definition \ref{def:crossingedges}  $e_2$ and $e_5$ are crossing. Here, we obtain a contradiction because by Definition \ref{def:comparison} the pair ($e_2$, $e_5$) is less than the pair($e_1, e_2$) and it contradict with the fact that ($e_1, e_2$) is the minimum crossing flow.
	\end{itemize}

	 In this function, to resolve this cross the amount of flow crossing through both edges  $e_1$ and $e_2$ decreases by the minimum of them. This amount is added to the flow crossing through edges $e_3=(\vertex{t}{x}{y_1}, \vertex{t}{x+1}{y'_2})$ and $e_4=(\vertex{t}{x}{y'_1},\vertex{t}{x+1}{y_2})$. After applying the function the mentioned cross is resolved hence no flow passes through at least one of the edges forming it, but it is possible to have new crossing flows since we add flow to at most two edges that there was no flow crossing through them. We explore the new possible crosses that edges $e_3$ and $e_4$ form, separately. 
	\begin{itemize}
	\item $e_3$: If before resolving the mentioned cross, $f_t(e_3)>0$ holds, adding flow to it does not form a new crossing flow. So we assume that $f_t(e_3)=0$ and explore the possible new crosses after adding flow to it. Recall that by Definition \ref{def:crossingedges} the edge $e_5=(\vertex{t}{x}{y"_1}, \vertex{t}{x+1}{y"_2})$ crosses  $e_3$ iff ($y"_1<y_1$ and $y'_2<y"_2$) or ($y_1<y"_1$ and $y"_2<y'_2$), but we have proved that both of these conditions contradict with the fact that $(e_1, e_2)$ is the minimum crossing flow. So increasing the $f_t(e_3)$ does not form a new cross.

	\item $e_4$: The same as $e_3$, we assume that the initial flow of this edge is zero. Recall that the edge $e_5=(\vertex{t}{x}{y"_1}, \vertex{t}{x+1}{y"_2})$ crosses  $e_4$ iff ($y"_1<y'_1$ and $y_2<y"_2$) or ($y'_1<y"_1$ and $y"_2<y_2$). Since we have proved that the relations ($y_1 < y"_1$ and $y"_2 \leq y'_2$) and ($y"_1\leq y_1$ and $y'_2< y"_2$) are invalid and $y_1<y'_1$ and $y'_2<y_2$, the relations ($y'_1 < y"_1$ and $y"_2 \leq y'_2$) and ($y"_1\leq y_1$ and $y_2< y"_2$) are also invalid. So the following condition holds for $e_5$:
	$$ (y_1<y"_1<y'_1 \text{ and } y_2<y"_2) \text{ or } (y'_1<y"_1 \text{ and } y'_2<y"_2<y_2) $$ 
		One can easily see that by Definition \ref{def:comparison} the crossing flow $(e_4,e_5)$ is greater than the crossing folw $(e_1,e_2)$. 
	\end{itemize}
	Hence after resolving the crossing pair $(e_1, e_2)$ both the edges $e_3$ and $e_4$ does not form a cross less than the $(e_1, e_2)$, the minimum cross gets greater at each call of the function \textproc{ResolveCross}, and since there are polynomially many number of possible crossing edges, the runtime of the function \textproc{ResolveCrosses} is polynomial.
\end{proof}

\subsection*{Theorem~\ref{th:compatibility}}
\xhdr{Statement.}

Let $f = \langle f_1, \ldots, f_\time \rangle$ be a solution of LP~\ref{lp:lp} without any crossing flows. There exists a polynomial time algorithm to find a mixed strategy of the defender that is equivalent to $f$.

\begin{proof}
Let $p_1, \ldots, p_\time$ be the top-most flow paths of $f_1, \ldots, f_\time$ respectively. By Lemma~\ref{lem:topmost-compatible}, any two consecutive top-most flow paths $p_i$ and $p_{i+1}$ are compatible. Therefore there exists a valid pure strategy $p$ of the defender, such that the placement of patrols at time point $t$ in $p$, is the same as the equivalent pure snapshot of $p_t$. We use an iterative algorithm to construct the desired mixed strategy $s$:
\begin{enumerate}
	\item Find the top-most flow paths $p_1, \ldots, p_\time$ of $f_1, \ldots, f_\time$.
	\item Construct the pure strategy $p$, corresponding to $p_1, \ldots, p_\time$.
	\item Add $p$ to $s$ with probability $q = \min \sizeofflowpath{p_i}$.
	\item For any edge $e$ of any $p_i$, decrease $\edgeflow{f_i}{e}$ to $\edgeflow{f_i}{e} - q$.
	\item If there is any flow left in $f_1, \ldots, f_\time$, repeat all the steps.
\end{enumerate}

\textbf{Correctness}: Note that, initially $f_1, \ldots, f_\time$ are flows of size 1, therefore $s$ is indeed a probability distribution over some pure strategies, and thus is a mixed strategy. In addition, although we change $f_1, \ldots, f_\time$ at each step, but the compatibility argument of top-most flow paths still holds. Let $g = \langle g_1, \ldots, g_\time \rangle$ denote the remaining flows after step 4 of the algorithm. We first prove that after decreasing the flow of some edges in step 4 the following condition, which is a constraint of the LP, also holds for $g$:
\begin{equation} \label{eq:g-lp}
		\Sigma_{v \in \vertexcons{t}{k}{i}{j}}\outflow{g_t}{v} \leq \Sigma_{v' \in \outfeas{t}{\vertexcons{t}{k}{i}{j}}} \outflow{g_{t+1}}{v'}  \quad \forall t,k,i, j: t\in\timeset, k\in \Kset, 1 \leq i \leq j \leq \parsize{t}
\end{equation} 
We first assume there exist valid values for $t$, $i$, $j$ and $k$ such that the following holds:

\begin{equation} \label{eq:11}
		\Sigma_{v \in \vertexcons{t}{k}{i}{j}}\outflow{g_t}{v} > \Sigma_{v' \in \outfeas{t}{\vertexcons{t}{k}{i}{j}}} \outflow{g_{t+1}}{v'}  
\end{equation} 
Then we obtain a contradiction. Let \vertex{t}{k}{x} and \vertex{t+1}{k}{y} respectively denote the $k$-th vertices of the top-most paths in $f_t$ and $f_{t+1}$. Holding Equation \ref{eq:11} yields both $\vertex{t}{k}{x} \notin \vertexcons{t}{k}{i}{j}$ and $\vertex{t+1}{k}{y} \in \outfeas{t}{\vertexcons{t}{k}{i}{j}}$. Also, One can easily see that $j<x$. Moreover, since $\vertex{t+1}{k}{y} \in \outfeas{t}{\vertexcons{t}{k}{i}{j}}$, and there exists no flow passing through vertices above $\vertex{t+1}{k}{y}$ in both $g_t$ and $f_t$,
	\begin{equation}	\label{eq:13}
	 \Sigma_{v' \in \outfeas{t}{\vertexcons{t}{k}{i}{x}}} \outflow{f_{t+1}}{v'}= \Sigma_{v' \in \outfeas{t}{\vertexcons{t}{k}{i}{j}}} \outflow{g_{t+1}}{v'} + q
	 \end{equation}
	In addition, hence  $\vertex{t}{k}{x} \notin \vertexcons{t}{k}{i}{j}$, 
	\begin{equation} \label{eq:14}
	\Sigma_{v \in \vertexcons{t}{k}{i}{j}}\outflow{g_t}{v} +q \leq \Sigma_{v \in \vertexcons{t}{k}{i}{x}}\outflow{f_t}{v} 
	\end{equation} 
	So, by (\ref{eq:13}) and (\ref{eq:14}): 
	\begin{equation} \label{eq:15}
		\Sigma_{v \in \vertexcons{t}{k}{i}{j}}\outflow{g_t}{v} + \Sigma_{v' \in \outfeas{t}{\vertexcons{t}{k}{i}{x}}} \outflow{f_{t+1}}{v'} \leq \Sigma_{v \in \vertexcons{t}{k}{i}{x}}\outflow{f_t}{v} + \Sigma_{v' \in \outfeas{t}{\vertexcons{t}{k}{i}{j}}} \outflow{g_{t+1}}{v'}
	 \end{equation}
	Having both (\ref{eq:11}) and (\ref{eq:15}) yields the following inequality: 
	\begin{equation}	
	 \Sigma_{v' \in \outfeas{t}{\vertexcons{t}{k}{i}{x}}} \outflow{f_{t+1}}{v'} <
	\Sigma_{v \in \vertexcons{t}{k}{i}{x}}\outflow{f_t}{v} 
	\end{equation} 
	Note that $f$ is a solution of the LP, and this is a contradiction with the line \ref{cons:compatibility} in LP~\ref{lp:lp}. So, equation \ref{eq:g-lp} holds for $g$. 
	
	However, yet $g$ is not a set of canonical flows since for any $t\in \timeset$ the amount of flow in $g_t$ is equal to $1-q$. To resolve this we multiply the flow in all the edges by $\frac{1}{1-q}$. So, $g$ is a collection of canonical paths. 	Note that the top-most paths are unchanged by this multiplication and equation \ref{eq:g-lp} still holds for $g$. Recall that by Lemma \ref{lem:topmost-compatible} the consecutive top most paths in $g$ are compatible, so the correctness of this algorithm is proved.
	
\textbf{Running Time}: The algorithm will halt in polynomial rounds since at each round the flow passing through at least one edge in some flow $f_i$ will be decreased to zero and the number of edges is polynomial.
\end{proof}

\end{document}